\documentclass[review]{elsarticle}

\usepackage{graphicx}      
\usepackage{natbib}        
\usepackage{amssymb,mathtools,amsfonts,bm}
\usepackage{subfigure}
\usepackage{verbatim}
\usepackage{balance}
\usepackage{appendix}
\usepackage{algorithm,algorithmic}
\usepackage{booktabs}

\usepackage{float}
\usepackage[normalem]{ulem}
\usepackage[usenames, dvipsnames]{color}
\usepackage{setspace}
\usepackage{graphics}
\usepackage{hyperref}
\usepackage{subfigure}
\usepackage{verbatim}
\usepackage{amsmath}
\usepackage{amssymb}
\usepackage[T1]{fontenc}
\usepackage{calc}
\usepackage{soul}
\usepackage{flushend}
\usepackage{epstopdf} 
\usepackage{picins}

\def\maxlink{\mathop{\mathrm{max-link}}}
\def\iid{\mathop{\mathrm{i.i.d.}}}
\def\minpow{\mathop{\mathrm{min-power}}}
\def\minpow{\mathop{\mathrm{min-power}}}

\def\SR{\mathop{\{S\hspace{-0.09cm} \rightarrow \hspace{-0.09cm} R\}}}
\def\RD{\mathop{\{R\hspace{-0.09cm} \rightarrow \hspace{-0.09cm} D\}}}
\def\RR{\mathop{\{R\hspace{-0.09cm} \rightarrow \hspace{-0.09cm} R\}}}
\def\SD{\mathop{\{S\hspace{-0.09cm} \rightarrow \hspace{-0.09cm} D\}}}
\def\TD{\mathop{\{T\hspace{-0.09cm} \rightarrow \hspace{-0.09cm} D\}}}

\def\trace{\mathop{\mathrm{tr}}}  
\newcommand\norm[1]{\left|#1\right|}

\newtheorem{proposition}{Proposition}

\newtheorem{remark}{Remark}

\newtheorem{proof}{Proof}

\journal{Ad Hoc Networks}

\newenvironment{list4}{
  \begin{list}{$\bullet$}{
      \setlength{\itemsep}{0.05cm}
      \setlength{\labelsep}{0.2cm}
      \setlength{\labelwidth}{0.3cm}
      \setlength{\parsep}{0in}
      \setlength{\parskip}{0in}
      \setlength{\topsep}{0in}
      \setlength{\partopsep}{0in}
      \setlength{\leftmargin}{0.17in}}}
      {\end{list}}

\begin{document}
\begin{frontmatter}

%
%
%
%
\title{Relay-Pair Selection in Buffer-Aided Successive Opportunistic Relaying using a Multi-Antenna Source}

\author[First]{Themistoklis~Charalambous}
\author[Second]{Su~Min~Kim\corref{mycorrespondingauthor}}
\author[Third]{Nikolaos~Nomikos}
\author[Fourth]{Mats~Bengtsson}
\author[Fifth]{Mikael Johansson}

\cortext[mycorrespondingauthor]{Corresponding author. This work was supported in part by the MSIT, Korea, under the ITRC support program(IITP-2018-0-01426) supervised by the IITP.}

\address[First]{Department of Electrical Engineering and Automation, Aalto University, Espoo, Finland. Email: {\tt  themistoklis.charalambous@aalto.fi}.}
\address[Second]{Department of Electronics Engineering, Korea Polytechnic University, Siheung, Korea. Email: {\tt suminkim@kpu.ac.kr}.}
\address[Third]{Department of Information and Communication Systems Engineering, University of the Aegean, Samos, Greece. Email: {\tt nnomikos@aegean.gr}.}
\address[Fourth]{Department of Information Science and Engineering, School of Electrical Engineering and Computer Science, KTH Royal Institute of Technology, Stockholm, Sweden. Email:  {\tt mats.bengtsson@ee.kth.se}.}
\address[Fifth]{Department of Automatic Control, School of Electrical Engineering and Computer Science, KTH Royal Institute of Technology, Stockholm, Sweden. Email:  {\tt mikaelj@kth.se}.}

%
%
%
%
\begin{abstract}
We study a cooperative network with a buffer-aided multi-antenna source, multiple half-duplex (HD) buffer-aided relays and a single destination. Such a setup could represent a cellular downlink scenario, in which the source can be a more powerful wireless device with a buffer and multiple antennas, while a set of intermediate less powerful devices are used as relays to reach the destination. The main target is to recover the multiplexing loss of the network by having the source and a relay to simultaneously transmit their information to another relay and the destination, respectively. Successive transmissions in such a cooperative network, however, cause inter-relay interference (IRI). First, by assuming global channel state information (CSI), we show that the detrimental effect of IRI can be alleviated by precoding at the source, mitigating or even fully cancelling the interference. A cooperative relaying policy is proposed that employs a joint precoding design and relay-pair selection. Note that both fixed rate and adaptive rate transmissions can be considered. For the case when channel state information is only available at the receiver side (CSIR), we propose a relay selection policy that employs a phase alignment technique to reduce the IRI. The performance of the two proposed relay pair selection policies are evaluated and compared with other state-of-the-art relaying schemes  in terms of outage and throughput. The results show that the use of a powerful source can provide considerable performance improvements.
\end{abstract}


\end{frontmatter}


%
%
%
%
\section{Introduction}

The communication paradigm of cooperative relaying has recently received considerable attention due to its effectiveness in alleviating the effects of multipath fading, pathloss and shadowing, and its ability to deliver improved performance in cognitive radio systems and wireless sensor networks. Relay-assisted cellular networks are a promising solution for enhancing coverage and are included in standards, such as IEEE 802.16j/m and 3GPP Long Term Evolution-Advanced (LTE-A).

 Several protocols for cooperative relaying were presented in \cite{LAN} where the gains in transmit and receive diversity were studied. In multi-relay networks, simultaneous transmissions by the relays are in general difficult to handle; towards this end, opportunistic relay selection has been suggested in \cite{BLE} to improve the resource utilization and to reduce the hardware complexity. Stemming from the relay selection concept, many studies have proposed improved selection techniques (see, \emph{e.g.}, \cite{BLE1, KAR, KRI1}).
Traditional HD relaying schemes partition the packet transmission slot into two phases, where the transmission on the Source-Relay $\SR$ link happens in the first phase, and the transmission on the Relay-Destination $\RD$ link occurs in the second phase. However, this relaying scheme limits the maximum achievable multiplexing gain to $0.5$, which also results in bandwidth loss. 

 In order to overcome such multiplexing and bandwidth limitations, several techniques have been proposed in the literature (see, for example, \cite{DIN} and references therein). Among them, the successive relaying scheme in \cite{FAN} incorporates multiple relay nodes and allows concurrent transmissions between source-relay and relay-destination to mimic an ideal full-duplex transmission. However, this scheme targets scenarios with a long distance between the relays and thus inter-relay interference is not considered. An extension of this work is discussed in \cite{chao_ISIT}, where the authors assume that IRI is strong (in co-located or clustered relays) and can always be decoded at the affected nodes; this decoded IRI is exploited in a superposition coding scheme that significantly improves the diversity-multiplexing trade-off performance of the system. 
 
In earlier studies, relays were assumed to lack data buffers and relay selection was mainly based on the $\max-\min$ criterion and its variations (see, for example, \cite{BLE,BLE1,KAR,KRI1}). Here, the relay that receives the source signal is also the one that forwards the signal to the destination. With the adoption of buffer-aided relays, this coupling is broken, since different relays could be selected for transmission and reception, thus allowing increased degrees of freedom. Buffering at the relay nodes has been shown to be a promising solution for cooperative networks and motivates the investigation of new protocols and transmission schemes (see \cite{survey} and \cite{ZLA_CM} for an overview). Ikhlef \emph{et al.} \cite{IKH2} proposed a novel criterion based on $\max-\max$ relay selection (MMRS), in which the relay with the best source-relay $\SR$ link is selected for reception and the relay with the best relay-destination $\RD$ link is selected for transmission on separate slots.  In \cite{KRICHAR}, at each slot the best link is selected among \emph{all} the available $\SR$ and $\RD$ links, as a part of the proposed $\maxlink$ policy, thus offering an additional degree of freedom to the network, while buffer-aided link selection was studied in topologies with source-destination connectivity (see, \emph{e.g.,} in \cite{CHAR_CL, SHAQ_TWC}), resulting in improved diversity and throughput. More recently, to alleviate the throughput loss of HD relaying, significant focus has been given to minimizing the average packet delay. Towards this end, there have been proposed various approaches: hybrid solutions combining the benefits of MMRS and $\maxlink$ \cite{OIWA_TVT}, {the use of broadcasting in the $\SR$ link \cite{OIWA_ACCESS, NOM_TCOM2018},} the prioritization of $\RD$ transmissions \cite{NOM_TCOM2018,POUL1,TIAN_TVT} or the selection of the relay with the maximum number of packets residing in its buffer \cite{LIN_TVT}. However, it was shown that a trade-off exists between delay performance and the diversity of transmission as the number of relays with empty or full buffers increases. In \cite{POUL2} delay- and diversity-aware relay selection policies were proposed aiming at reducing the average delay by incorporating the buffer size of the relay nodes into the relay selection process.  

To improve throughput and reduce average packet delays, in a number of studies, it was proposed to employ non-orthogonal successive transmissions by the source and a selected relay. In order to recover the HD multiplexing loss, \cite{IKH3} suggests to combine MMRS with successive transmissions (SFD-MMRS). As the proposed topology aims to mimic full-duplex relaying, different relays are selected in the same time slot; however, relays are assumed to be isolated and the effect of IRI is ignored. More practical topologies were studied in \cite{NOM2, NOM_TCOM, SIM_TWC} where IRI exists and is not always possible to be cancelled. For fixed rate transmission, the proposed $\minpow$ solution proposed in \cite{NOM_TCOM}, combining power adaptation and interference cancellation, provides a performance close to the upper bound of SFD-MMRS. In addition, Kim and Bengtsson~\cite{KIM,KIM2} proposed buffer-aided relay selection and beamforming schemes taking the IRI into consideration; they consider a model which can be regarded as an example of relay-assisted device-to-device (D2D) communications, where the source and destination are low-cost devices with a single antenna and the relays comprise more powerful relays with buffers and multiple antennas. Numerical results show that their approach outperforms SFD-MMRS when interference is taken into consideration, and when the number of relays and antennas increases they approach the performance of the interference-free SFD-MMRS, herein called the \emph{ideal} SFD-MMRS. Finally, the use of buffer aided relays has been considered in topologies with full-duplex (FD) capabilities \cite{survey, NOM_FD1, ZLA_CL2016} and decoupled non-orthogonal uplink and downlink transmissions \cite{LIU_TCOM}, showing that throughput can be significantly improved due to the increased scheduling flexibility. Recently, in line with non-orthogonal successive transmissions, there has been considerable attention on non-orthogonal multiple access due to its superior spectral efficiency, and its combination with buffer-aided relaying was inevitable \cite{ZHANG_TCOM2017, NOM_GCW, CAO_TMC2018}.

In many practical considerations (\emph{e.g.}, wireless sensors), the relay nodes are hardware-limited to be HD while the source can be a more powerful wireless device with large buffers and multiple antennas. Although this observation is not always true (\emph{e.g.}, in D2D communications~\cite{KIM,KIM2}), it is a reasonable and common practical scenario. Towards this end, we study a network which consists of a buffer-aided source with a single or multiple antennas, multiple HD buffer-aided relays and a destination. With this setup, 
\begin{list4}
\item we are able to approach (achieve) the performance of the ideal SFD-MMRS by adopting a buffer-aided source whose buffer retains replicas of the relay buffers, successfully transmitted data from the source to relays, to facilitate IRI mitigation (cancellation); 
\item we relax the assumption of knowing the full CSI and instead, we allow for CSIR and limited feedback from the receiver to the transmitter; under these conditions we propose a relay pair selection scheme that is based on partial phase alignment of signals.
\end{list4}

The rest of the paper is organized as follows. In Section~\ref{sec:model}, the system model is outlined. The proposed relaying schemes for variable and fixed rate are presented in Sections~\ref{sec:policy1} and \ref{sec:policy2}, respectively. The performance of the proposed relaying policies in terms of outage and average throughput, along with comparisons with other state-of-the-art relaying schemes are presented in Section~\ref{sec:numerical}. Finally, conclusions and a discussion on future possible directions are drawn in Section~\ref{sec:conclusions}.

%
%
\subsubsection*{Notation}\label{sec:notation}

Vectors are written in bold lower case letters and matrices in bold capital letters. $\mathbb{R}$, $\mathbb{C}$ and $\mathbb{N}$ denote the sets of real numbers, complex and natural numbers, respectively. For a matrix~$\mathbf{A}\in \mathbb{C}^{n\times n}$ ($n\in \mathbb{N}$), $a_{ij}$ denotes the entry in row $i$ and column~$j$. Matrix $\mathbf{I}$ denotes the identity matrix of appropriate dimensions. The trace of a square matrix $\mathbf{A}\in\mathbb{R}^{n\times n}$ is denoted by $\trace(\mathbf{A})$. $(\cdot)^T$ and $(\cdot)^H$ denote the transpose and hermitian transpose operations, respectively; $\|\cdot \|$ denotes the $2$-norm operation. The complex conjugate of a complex number $z$ is denoted by $\overline{z}$; the real and imaginary parts of a complex number $z$ are denoted by $\mathcal{R}(z)$ and $\mathcal{I}(z)$, respectively. We denote by $\mathcal{P}$ the set of all possible relay-pairs in the relay network, and by $|\mathcal{P}|$ its cardinality. A relay pair, denoted by $(R,T)$ belongs to set $\mathcal{P}$ if and only if $R\neq T$, the receiving relay $R$ is not full and the transmitting relay $T$ is not empty. 

%
%
\section{System model}\label{sec:model}

We consider a cooperative network consisting of a buffer-aided source $S$ with multiple antennas, a set $\mathcal{K} \triangleq \{1, 2, \ldots, K \}$ of $K$ HD decode-and-forward (DF) relays with buffers, and a single destination $D$. Fig.~\ref{fig:model} illustrates a simple example with a buffer-aided two-antenna source $S$, two buffer-aided single-antenna HD DF relays, and a single-antenna destination $D$. To simplify the analysis, we examine the case where connectivity between the source and the destination is established only via the relays and ignore the direct $\SD$ link  (as in, \emph{e.g.}, \cite{IKH2,KRICHAR,IKH3,KIM,KIM2}).

The number of data elements in the buffer of relay $R_k$ is denoted by $Q_k$ and its capacity by $Q_{\max}$. For the fixed rate transmission policy, we assume that the packets are transmitted at a fixed rate of $C_0$ bits per channel use (BPCU), and the data of each transmission occupies a single time slot in the buffer. 

First, we provide the signals received at relay $R$ and destination $D$. At the destination, at any arbitrary time-slot $n$ the following signal is received:
\begin{align}
y_{D}[n]=h_{TD}x[p]+w_D[n] \; ,
\end{align}
where $x[p]$ is the signal received and stored  in a previous time-slot $p$ in the buffer of the now transmitting relay $T$, $h_{TD}$ denotes the channel coefficient from the transmitting relay $T$ to the destination $D$, and $w_{D}[n]$ denotes the additive white Gaussian noise (AWGN) at the destination in the $n$-th time slot, i.e., $w_{D}[n]\sim\mathcal{CN}(0,{\sigma_D^2})$. It must be noted that $x[p]$ was not necessarily received in the just before time-slot (i.e., $p\leq n-1$). At the same time, the reception of the source's signal by relay $R$ is interfered from the transmission of $T$ which forwards a previous signal $x[p]$ to the destination. Thus, $R$ receives
\begin{align}\label{eq:yRgeneral}
y_{R}[n]=\sum_{i\in \mathcal{A}} h_{S_i R} x_{S_i}[n] + h_{TR} x_T[n]  +w_R[n] ,
\end{align}
where $\mathcal{A}$ denotes the index set of transmit antennas at the source, i.e., $\mathcal{A}=\{1,2,\ldots,\nu\}$, $h_{S_i R}$ denotes the channel coefficient from the $i$-th transmit antenna at the source to the receiving relay $R$, $h_{TR}$ denotes the channel coefficient from the transmitting relay $T$ to the receiving relay $R$, $x_{S_i}[n]$ denotes the transmitted signal from the $i$-th transmit antenna at the source in the $n$-th time slot, $x_T[n]$ denotes the transmitted signal from the transmitting relay in the $n$-th time slot (i.e., $x_T[n] = x[p]$), and $w_{R}[n]$ denotes the AWGN at the receiving relay in the $n$-th time slot, i.e., $w_{R}[n]\sim\mathcal{CN}(0,{\sigma^2})$.

\begin{figure}[t]
\centering
\includegraphics[width=0.75\columnwidth]{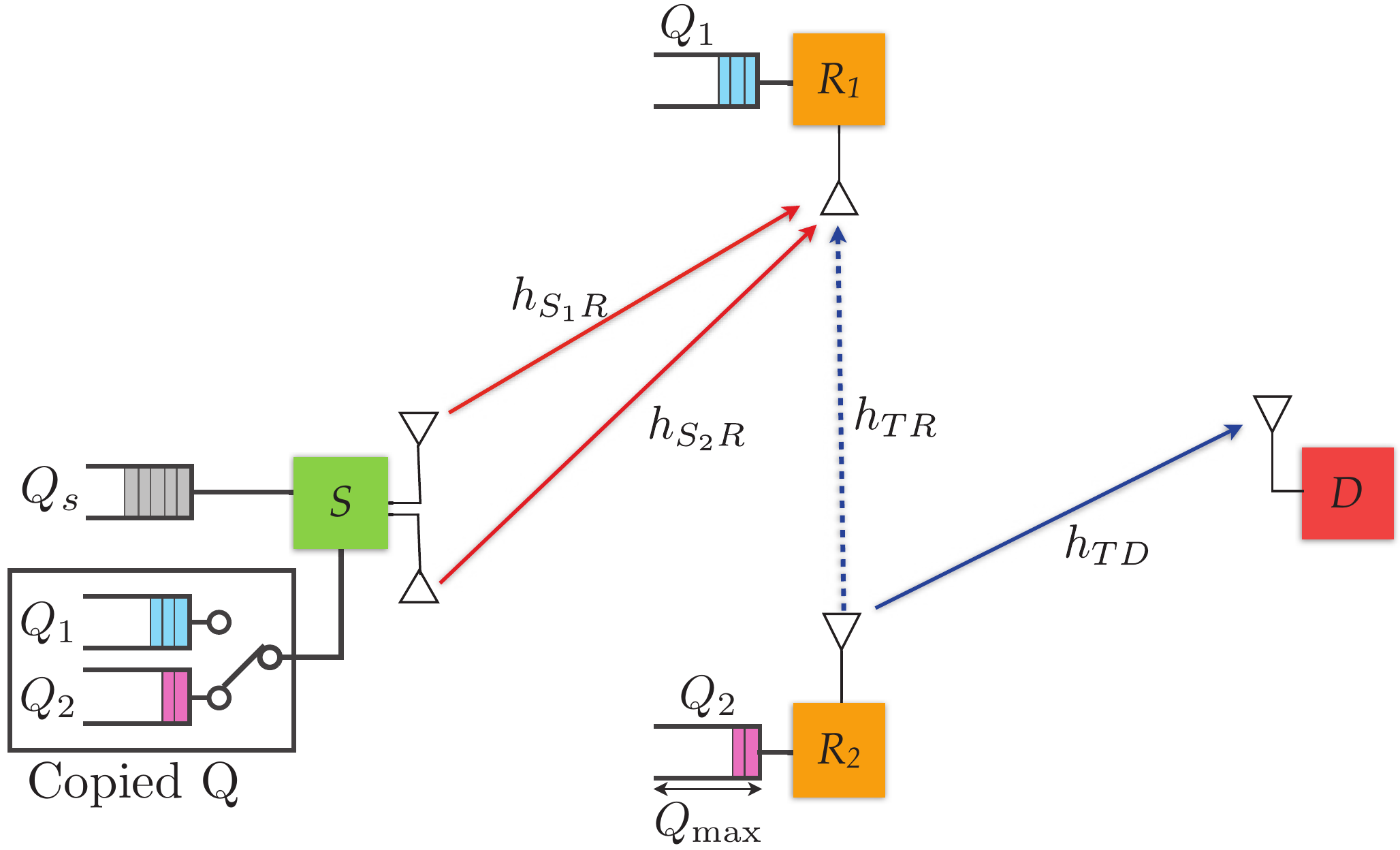}
\caption{A simple example of a cooperative network consisting of a {buffer-aided} source $S$ with two antennas ($S_1$ and $S_2$), two HD relays (the receiving relay is denoted by $R$ and the transmitting relay by $T$) and a destination $D$; in this example, $R_1 \equiv R$ and $R_2 \equiv T$. The buffers at $S$ consist basically of replicas of the data queues of the relays; the source has new packets in the source buffer $Q_{S}$ and replicas of the successfully transmitted packets to the relays in a set of copied buffers.}
\label{fig:model}
\end{figure}

The source $S$ is assumed to be saturated (infinite backlog) and hence, it has always data to transmit. The buffering memory at the source is organized into $K$ queues, that basically contain replicas of the data queues of the relays, in order to exploit it for IRI mitigation or cancellation.

The operation is assumed to be divided into time slots. In each time slot, the source and a relay simultaneously transmit their own data to mimic FD relaying (cf. \cite{IKH3,KIM,KIM2,NOM2,NOM_TCOM}). The transmission powers of the source and the $k$-th relay are denoted by $P_S$ and $P_T$, respectively. For notational simplicity, we assume throughout this paper that all devices use a common fixed transmit power level (i.e., $P_S=P_T=P,~\forall T\in\mathcal{K}$), unless otherwise specified. Moreover, we assume that the receivers send short-length error-free acknowledgment/negative-acknowledgment (ACK/NACK) messages over a separate control channel.

We assume narrowband Rayleigh block fading channels. Each channel coefficient is constant during one time slot and varies independently between time slots. For each time slot, the channel coefficient $h_{ij}$ for link $\{i\hspace{-0.09cm}\rightarrow \hspace{-0.09cm} j\}$  follows a circular symmetric complex Gaussian distribution with zero mean and variance $\sigma_{ij}^2$, i.e., $h_{ij}\sim\mathcal{CN}(0,\sigma_{ij}^2)$. Thus, the channel power gain $g_{ij} \triangleq |h_{ij}|^2$  follows an exponential distribution, i.e., $g_{ij}\sim\mathrm{Exp}({\sigma_{ij}^{-2}})$. 
In addition, we assume the AWGN at each receiver with variance $\sigma^2$.

%
%
%
%
\section{Buffer-Aided Relay Selection Based on Buffer-Aided Source Precoding}\label{sec:policy1}
The goal is to design a precoding matrix at the source such that, at each time slot $n$  the $\SR$ link for the relay selected to receive a packet from the source maximizes its SINR by transmitting a linear combination of (a) the source's desired signal for the receiving relay $R$ and (b) the signal of the transmitting relay $T$ taken out from the copied buffer at the source. Note that maximizing the SINR of the link is equivalent to minimizing the outage probability of the link. Thus, this precoding design criterion is relevant for both fixed rate and adaptive rate transmissions. 
 
\noindent \textbf{Special case: Source with a single antenna.} The received signal at the receiving relay for the $n$-th time slot, $y_R[n]$, can be expressed as
\begin{align}\label{signals_singleantenna}
y_R[n] &= h_{SR}\underbrace{\big( m_{1}x_S[n]+m_{2}x_T[n]\big)}_{\text{source signal}}+h_{TR} x_T[n] +w_R[n] \nonumber \\
&= h_{SR} \underbrace{\begin{bmatrix}m_1 & m_2\end{bmatrix}}_{\text{precoding matrix}} \underbrace{\begin{bmatrix}x_S[n] & x_T[n]\end{bmatrix}^T}_{\text{data from the source}} +h_{TR} x_T[n] +w_R[n],
\end{align}
where $x_{S_i}[n]$ in Eq.~\eqref{eq:yRgeneral} is written as $m_{1}x_S[n]+m_{2}x_T[n]$, and $m_{1},m_{2} \in \mathbb{C}$. Each transmitting node uses a fixed power $P$ and thus power control issues are not taken into account. The SINR at the receiving relay, denoted by $\Gamma_R$, is given by 
\begin{align}
\Gamma_R\triangleq \frac{|h_{SR} {m}_1|^2 P}{|h_{SR} m_2 + h_{TR}|^2 P + \sigma^2} .
\end{align}
Hence, the following optimization problem can be formulated:
\begin{subequations}
\label{OPT0_singleantenna}
\begin{align}
\displaystyle \underset{m_1,m_2}{\max}  & ~  \frac{|h_{SR} {m}_1|^2 P}{|h_{SR} m_2 + h_{TR}|^2 P + \sigma^2}\label{0obj1_sa}  \\[0.2cm]
\textrm{s.t.} 
& ~ |{m}_1|^2 + |{m}_2|^2 \leq 1. \label{0cond2_sa} 
\end{align}
\end{subequations}
The solution to optimization problem \eqref{OPT0_singleantenna} will be given for the general case and the solution to this special case can be found by simplifications \emph{mutatis mutandis}.

\noindent \textbf{General case: Source with multiple antennas.} Now, the precoding matrix at the source comprising $\nu$ antennas is $\mathbf{M}\in \mathbb{C}^{\nu \times 2}$. Matrix $\mathbf{M}$ has dimensions $\nu \times 2$, since the transmitted signals are formed as linear combinations of the source signal and the signal transmitted by the active relay. The received signal at the receiving relay for the $n$-th time slot, $y_R[n]$, can be expressed as
\begin{align}\label{signals}
y_R[n]& = \sum_{i\in \mathcal{A}}  h_{S_i R}\big( m_{i1}x_S[n]+m_{i2}x_T[n]\big)   +h_{TR} x_T[n] +w_R[n],
\end{align}
where $x_{S_i}[n]$ in Eq.~\eqref{eq:yRgeneral} is written as $m_{i1}x_S[n]+m_{i2}x_T[n]$, and $m_{i1},m_{i2} \in \mathbb{C}$. 
Defining 
\begin{align*}
\mathbf{h}_S^H &\triangleq \begin{bmatrix}h_{S_1R} & h_{S_1R} & \ldots & h_{S_{\nu}R}\end{bmatrix}, \\ 
\mathbf{m}_1 &\triangleq \begin{bmatrix}m_{11} & m_{21} & \ldots & m_{{\nu}1}\end{bmatrix}^T, \\ 
\mathbf{m}_2 &\triangleq \begin{bmatrix}m_{12} & m_{22} & \ldots & m_{{\nu}2}\end{bmatrix}^T, 
\end{align*}
Eq.~\eqref{signals} can be written as
\begin{align}\label{signals1}
y_R[n] &= \mathbf{h}_S^H \begin{bmatrix}\mathbf{m}_1 & \mathbf{m}_2\end{bmatrix} \underbrace{\begin{bmatrix}x_S[n] & x_T[n]\end{bmatrix}^T}_{\text{data from the source}}+h_{TR} x_T[n] +w_R[n].
\end{align}
\noindent From Eq.~\eqref{signals1}, it can be easily seen that the source node with $\nu$ antennas applies a linear precoding matrix $\mathbf{M}\triangleq \begin{bmatrix}\mathbf{m}_1 & \mathbf{m}_2\end{bmatrix}$ on the transmitted signals such that the IRI is reduced, while at the same time the signal from the source to the intended relay is increased. Each transmitting node uses a fixed power $P$ and thus power control issues are not taken into account. The SINR at the receiving relay, denoted by $\Gamma_R$, is given by 
\begin{align}
\Gamma_R\triangleq \frac{\mathbf{h}_S^H \mathbf{m}_1 \mathbf{m}_1^H \mathbf{h}_S P}{|\mathbf{h}_S^H \mathbf{m}_2 + h_{TR}|^2 P + \sigma^2} .
\end{align}
Hence, the following optimization problem can be formulated:
\begin{subequations}
\label{OPT0}
\begin{align}
\displaystyle \underset{\mathbf{m}_1,\mathbf{m}_2}{\max}  & ~  \frac{\mathbf{h}_S^H \mathbf{m}_1 \mathbf{m}_1^H \mathbf{h}_S P}{|\mathbf{h}_S^H \mathbf{m}_2 + h_{TR}|^2 P + \sigma^2} \label{0obj1}  \\[0.2cm]
\textrm{s.t.} 
& ~ \mathbf{m}_1^H \mathbf{m}_1 + \mathbf{m}_2^H \mathbf{m}_2 \leq 1. \label{0cond2} 
\end{align}
\end{subequations}

\begin{proposition}\label{propA}
The precoding matrix $\mathbf{M}$ that solves optimization problem~\eqref{OPT0} is given by
\begin{align}
\mathbf{M} = & \frac{1}{\| \mathbf{h}_S \|^2} \begin{bmatrix}
\sqrt{\| \mathbf{h}_S \|^2 - \omega^2 | h_{TR} |^2} \mathbf{h}_S,  &
-\omega h_{TR} \mathbf{h}_S  \end{bmatrix} \notag \\ 
= & \frac{\mathbf{h}_S}{\| \mathbf{h}_S \|^2} \begin{bmatrix}
\sqrt{\| \mathbf{h}_S \|^2 - \omega^2 | h_{TR} |^2},  & -\omega h_{TR}  \end{bmatrix} ,\label{M:sol}
\end{align}
where 
\vspace{-0.4cm}
{\small
\begin{align}\label{omega:sol}
\omega = \left( \frac{1}{2} + \frac{\| \mathbf{h}_S \|^2 + \rho}{2| h_{TR} |^2} \right) - \sqrt{\left( \frac{1}{2} + \frac{\| \mathbf{h}_S \|^2 + \rho}{2| h_{TR} |^2} \right)^2 - \frac{\| \mathbf{h}_S \|^2}{| h_{TR} |^2} }
\end{align}}
and $\rho=\sigma^2/P$.
\end{proposition}

\begin{proof}
See appendix~\ref{A:proof_propA}.
\end{proof} 
The last expression in \eqref{M:sol} shows that the optimal $M$
is a combination of standard maximum ratio transmission across the
source antennas, and the single antenna interference suppression solution.

Proposition~\ref{propA} shows that the source does not need to cancel the interference completely in order to achieve the maximum SINR at the receiving relay. This result appears counter-intuitive, since interference is the main source of performance degradation and the SINR should be higher when the IRI is cancelled. However, this occurs due to the limited power at the source. When the power $P$ is high enough, $\rho$ gets very small. Hence,  from \eqref{omega:sol},
\begin{align*}
\lim_{\rho\rightarrow 0} \omega &= \left( \frac{1}{2} + \frac{\| \mathbf{h}_S \|^2}{2| h_{TR} |^2} \right) - \sqrt{\left( \frac{1}{2} + \frac{\| \mathbf{h}_S \|^2}{2| h_{TR} |^2} \right)^2 - \frac{\| \mathbf{h}_S \|^2}{| h_{TR} |^2} } \\
& =
\begin{cases}
1, &\text{if }\| \mathbf{h}_S \|^2 \geq | h_{TR} |^2, \\[0.2cm]
\displaystyle \frac{\| \mathbf{h}_S \|^2}{| h_{TR} |^2}, &\text{if }\| \mathbf{h}_S \|^2 < | h_{TR} |^2.
\end{cases}
\end{align*}
which means that interference is essentially cancelled if $\|
\mathbf{h}_S \|^2 \geq | h_{TR} |^2$. 
If the source can change its power, then the power levels of the source and the transmitting relay are not necessarily the same. In what follows, we allow the source to adjust its power. The SINR now becomes
\begin{align}
\displaystyle \frac{\mathbf{h}_S^H \mathbf{m}_1 \mathbf{m}_1^H \mathbf{h}_S P_S}{|\sqrt{P_S}\mathbf{h}_S^H \mathbf{m}_2 + \sqrt{P_T}h_{TR}|^2  + \sigma^2}. \label{eq:SINR1}
\end{align}
The following proposition gives the optimal 
joint selection of $\omega$ and $P_S$, that maximizes the SINR, as given in \eqref{eq:SINR1}.

\begin{proposition}\label{propB}
When $P$ and $\mathbf{M}$  are considered jointly, then the optimal power $P_S^*$ and $\omega^*$ are given by
\begin{subequations}
\begin{align}
P_S^*&=P_{\max}, \\
\omega^* &= \min\left\{1,{\frac{\|\mathbf{h}_S\|}{| h_{TR} |}}, \omega_1, \omega_2  \right\} ,
\end{align}
where
\begin{align}
\omega_1 &= \frac{P_T| h_{TR} |^2 + \sigma^2}{ \sqrt{P_T P_{\max}}| h_{TR} |^2},\\
\omega_2 &=\frac{P_{\max} \|\mathbf{h}_S\|^2 + P_T| h_{TR} |^2+\sigma^2}{2 \sqrt{P_T P_{\max}|} h_{TR} |^2 } - \nonumber \\
& \hspace{-0cm} \frac{\sqrt{(P_{\max} \|\mathbf{h}_S\|^2 + P_T| h_{TR} |^2+\sigma^2)^2 - 4P_T P_{\max}\|\mathbf{h}_S\|^2 | h_{TR} |^2 }}{2 \sqrt{P_T P_{\max}|} h_{TR} |^2 }.
\end{align}
\end{subequations}
\end{proposition}
\begin{proof}
See appendix~\ref{B:proof_propB}.
\end{proof} 

Proposition~\ref{propB} suggests 
that in the case which $P_S$ and $\omega$ are jointly optimized, in order to maximize the SINR, the source transmits with $P_{\max}$ and chooses the $\omega$ corresponding to $P_{\max}$.

%
%
\subsection{Adaptive Rate Relay Pair Selection Policy}\label{sec:AR_RS_policy}

Since CSIs of $\RD$ links are assumed to be available at the source and relays, adaptive rate transmission, which implies transmission rate is determined by instantaneous CSI at each link, is possible. In this case, the source encodes its own desired codewords based on CSI of the selected $\SR$ link and codewords corresponding to the transmitting relay based on CSI of the selected $\RD$ link. Then, the source combines them via precoding. In the adaptive rate transmission, the main objective is to maximize the average end-to-end rate from the source to the destination, which is equivalently expressed as the average received data rate at the destination \cite{KIM2}, i.e.,
\begin{align}
\bar{C} = \underset{{W\to\infty}}{\lim} \frac{1}{W} \sum_{t=1}^{W}C_{T_t^* D}(t),
\label{eq:bar_C}
\end{align}
such that
\begin{align*}
\sum_{t=1}^{W}C_{T_t^* D}(t) \leq \sum_{t=1}^{W}C_{SR_t^*}(t), 
\end{align*}
where $W$ denotes a time window length and $C_{SR_t^*}(t)$ and $C_{T_t^*D}(t)$ denote the resulting instantaneous link rates depending on time slot $t$ once the relay pair $(R_t^*,T_t^*)$ is selected. The instantaneous rates are obtained by
\begin{align}
C_{S R}(t) &= \min\left\{ \log_2(1+\gamma_{S R}(t)), Q_{\max} - Q_{R}(t-1) \right\},\\
C_{T D}(t) &= \min\left\{ \log_2(1+\gamma_{T D}(t)), Q_{T}(t-1) \right\},
\end{align}
where $\gamma_{S R}(t)$ and $\gamma_{T D}(t)$ denote the effective SINR/SNR (signal-to-noise-ratio) of $\SR$ and $\TD$ links at time slot $t$ after applying the precoding matrix obtained by Proposition~\ref{propA}, i.e.,
\begin{align}
\gamma_{S R}(t) &= \frac{\left(\| \mathbf{h}_S \|^2 - \omega^2 | h_{TR} |^2\right)P}{(1-\omega^2)| h_{TR} |^2 P+\sigma^2} , \\[0.12cm]
\gamma_{T D}(t) &=  \frac{|h_{TD}|^2 P}{\sigma^2} ,
\end{align}
where $\omega$ is given in \eqref{omega:sol},
$Q_{\max}$ denotes the maximum length of queue, and $Q_{R}(t-1)$ and $Q_{T}(t-1)$ are the queue lengths of receiving and transmitting relays in BPCU at time slot $(t-1)$. For the selected relay pair $(R^*, T^*)$, the queue lengths are updated at the end of each time slot as
\begin{align}
Q_{R^*}(t) = Q_{R^*}(t-1) + C_{S R^*}(t),\\
Q_{T^*}(t) = Q_{T^*}(t-1) - C_{S T^*}(t).
\end{align}

Hereafter, based on the designed precoding matrix $\mathbf{M}$ in Section~\ref{sec:policy1}, we propose an adaptive rate relay selection policy for maximizing the average end-to-end rate. Applying a Lagrangian relaxation, the objective function to maximize the average end-to-end rate in \eqref{eq:bar_C} for the adaptive rate transmission becomes equivalent to a weighted sum of instantaneous link rates \cite{KIM,KIM2}.
The resulting relay selection scheme has the form
\begin{subequations}
\label{eq:RS_AR}
\begin{align}
\displaystyle \underset{(R, T)}{\max} &~ \delta_R C_{S R}(t) + (1-\delta_T) C_{T D}(t),\\
\textrm{s.t.} &~ R \neq T,~R,T \in \mathcal{K},
\end{align}
\end{subequations}
for predetermined weight factors\footnote{Under i.i.d. channel conditions, the weight factors are approximately identical for all relays and can be reduced as a single weight factor.} $\delta_k\in[0,1],~k\in\mathcal{K}$,
where $C_{S R}(t)$ and $C_{T D}(t)$ denote the instantaneous rates of $\SR$ and $\TD$ links at time slot $t$, respectively.
It is worth noting that for finite buffer size, the optimal weight factors can be found using either a subgradient method or a back-pressure algorithm \cite{KIM2}. According to the back-pressure algorithm, the weight factors can be found as $\delta_k = 1 - Q_k(t)/Q_S(t)$ during a training period where $Q_k(t)$ and $Q_S(t)$ denote the buffer occupancies at time $t$ at the $k$-th relay and the source, respectively. The details can be found in \cite{KIM2}.

\begin{remark}
The optimal relay pair selection in \eqref{eq:RS_AR} requires an exhaustive search with $K\times(K-1)$ combinations. Therefore, the computational complexity of optimal relay pair selection is $\mathcal{O}(K^2)$.
\end{remark}

\begin{remark}
The implementation of the proposed scheme requires global CSI of the instantaneous channels and the buffer states. The central unit (\emph{e.g.}, the source in this case) uses this information in order to select the appropriate relay pair. This global CSI requires a continuous feedback for each wireless link as we use a continuous monitoring of the ACK/NACK signaling in order to identify the status of the buffers. Although these implementation issues are beyond the scope of this paper, the proposed scheme can be implementable with the aid of various centralized and distributed CSI acquisition techniques and relay selection approaches as in \cite{BLE} where each relay sets a timer in accordance to the channel quality and through a countdown process the selection of the best relay is performed. Additionally, CSI overhead can be reduced significantly, through distributed-switch-and-stay-combining as in \cite{DSSC1, NOM_CAMAD, DSSC2}.
\end{remark}

%
%
%
%
\section{Buffer-Aided Relay Selection Based on Buffer-Aided Phase Alignment}\label{sec:policy2}

In this section, we relax the assumption of having knowledge of the full CSI. Instead, we allow for CSIR knowledge with limited feedback, i.e., each receiving node has CSI of the channel it receives data from and it can provide some information to the transmitting node. More specifically, each (receiving) relay feeds back to the source a phase value via a reliable communication link. 
This approach is closely related to phase feedback schemes proposed and standardized for MISO transmission (see, for example, \cite{P1, P2, P3, P4}).
The source signal can use this phase value in one of two possible ways: 
\begin{itemize}
\item[(a)] to mitigate the interfering signal, so that the overall interference is reduced or even eliminated; 
\item[(b)] to amplify the interfering signal, so that it can be decoded and removed from the rest of the received signals.
\end{itemize}
The phase value can be quantized into the desired number of bits, using uniform quantization.
We aim at recovering the multiplexing loss of the network by having the source and a relay to align phases, such that the IRI is reduced. As a result, the proposed phase alignment approach reduces the communication overhead, while introducing a smart quantization that further reduces the complexity of the transmission. 

Since there is no full CSI knowledge, as assumed in the scheme proposed in the previous section, no precoding can be applied and a relay-selection policy is proposed based on CSIR knowledge with limited feedback that aims at taking advantage of the multi-antenna source. Selection of the relay-pair is performed by choosing the pair that achieves the maximum end-to-end SINR. 

\subsection{Buffer-Aided Phase Alignment (BA-PA)}

Since the source has no CSI for any of the links between the antennas at the source and the relays, equal power allocation across all antennas is a natural approach at the transmitter side. However, for overhead reduction on CSI estimation at the receiver, we consider to use just two of them:
one for transmitting the packet to a relay and the other to mitigate the IRI. The existence of more than two antennas, however, can increase the diversity gain by choosing a subset of antennas based on CSI, while the fact that not all available antennas are included, the overhead for channel information is limited. The same principle can equally well be applied using a single antenna
source node, where one and the same antenna is used both for the
desired packet and the IRI mitigation.
Assuming two antennas used, the receiving signal at  $R$ in ~\eqref{eq:yRgeneral} is given by 
\begin{align}\label{eq:yR2}
y_{R}[n]&= h_{S_1 R} m_1 x_{S_1}[n] + h_{S_2 R} m_2 x_{S_2}[n] + h_{TR} x[p]  +w_R[n] ,
\end{align}
where we set $m_1=1/\sqrt{2}$ and $m_2=e^{j\phi}/\sqrt{2}$.
In each time slot, the signal from the second antenna of the source $x_{S_2} [n]$ is used in one of the following two ways:
\begin{itemize}
\item[{(a)}] to minimize the interference caused by the transmitting relay; this is done by transmitting $x[p]$ (i.e., $x_{S_2} [n]=x[p]$) with a shifted phase such that the interfering signal and the signal from the second antenna is in anti-phase. The optimal phase for this is given by 
\begin{align}\label{opt:1_1}
\phi^\star=\arg \min_{\phi} \norm{\frac{h_{S_2 R}}{\sqrt{2}}e^{j\phi} +h_{TR}}^2 ;
\end{align}
\item[{(b)}] to maximize the interference caused by the transmitting relay in order to make the signal strong enough to be decoded first, and hence, eliminate it. The optimal phase for this is given by 
\begin{align}\label{opt:1_2}
\phi^\dagger=\arg \max_{\phi} \norm{\frac{h_{S_2 R}}{\sqrt{2}}e^{j\phi} +h_{TR}}^2 .
\end{align}
\end{itemize}

\begin{proposition}\label{propC}
The phase $\phi^\star$ such that the interfering signal at the receiving relay $R_r$ from the transmitted signal $x[p]$ of the transmitting relay $R_t$ is minimized is given by
\begin{align}
e^{j\phi^\star} = -\frac{h_{S_2R}^H h_{TR}}{|h_{S_2 R}^H h_{TR}|} .
\end{align}
Similarly, the phase $\phi^\dagger$ such that the interfering signal at the receiving relay $R_r$ from the transmitted signal $x[p]$ of the transmitting relay $R_t$ is maximized is given by
\begin{align}
e^{j\phi^\dagger} = \frac{h_{S_2R}^H h_{TR}}{|h_{S_2R}^H h_{TR}|} .
\end{align}
\end{proposition}
\begin{proof}
See appendix~\ref{C:proof_propC}.
\end{proof}

Proposition~\ref{propC} gives the expressions for phase alignment for
each of the two approaches considered. By appropriately choosing the
phase shift $\phi$ of the signal from one of the source's antennas,
the source can minimize or maximize the interfering signal in order to
mitigate it or eliminate it completely. Note that the optimal value of
$\phi$ (either $\phi^\star$ or $\phi^\dagger$) can be quantized into
the desired number of bits, using uniform quantization, and be fed
back to the source. 

%
%
\subsection{Fixed Rate Relay Pair Selection Policy}\label{sec:FR_RS_policy}

Since only CSIR is available, the transmitters use fixed rate transmission. In such a case, the main objective is to minimize the outage probability. Independently of nodal distribution and traffic pattern, a transmission from a transmitter to its corresponding receiver is successful (error-free) if the SINR of the receiver is above or equal to a certain threshold, called the \textit{capture ratio}\footnote{It depends on the modulation and coding schemes as error-correction coding techniques supported by a wireless communication system and corresponds to the required SINR to guarantee the data rate of the application.} $\gamma_0$. Therefore, at the receiving relay $R$ for a successful reception when at the same time relay $T$ is transmitting we require that
\begin{align}\label{SINR_geq1}
\Gamma_{R}^{T}\triangleq \frac{|h_{S_1 R}|^2 P/{2}}{\left|\frac{h_{S_2 R}}{\sqrt{2}}e^{j\phi}+h_{TR}\right|^2 P + \sigma^2} \geq \gamma_0,
\end{align}
and at the destination we require that
\begin{align}\label{SNR_geq1}
\Gamma_{D}^T\triangleq \frac{|h_{TD}|^2 P}{\sigma^2} \geq \gamma_0.
\end{align}

An outage event occurs at the relay $R$ and destination $D$ when $\Gamma_{R}^S < \gamma_0$ and $\Gamma_{D}^T < \gamma_0$, respectively. The outage probability is denoted by $\mathbb{P}(\Gamma_{i}^k < \gamma_0)$, where $i$ represents the receiving node and $k$ the transmitting node.
Each relay $i$ is able to estimate the SINR for each transmitting relay $k$, denoted by $\Gamma_{i}^{k}$, $k\neq i$ (the full pilot protocol needed to the channel estimation is out of the scope of this work). We assume that this information can be communicated to the destination. In addition, the destination node can compute its own SNR due to each of the transmitting relays, denoted by $\Gamma_{D}^{k}$, $k\in\{1,\ldots, K\}$. Finally, we assume that the destination node has buffer state information\footnote{The destination can know the status of the relay buffers by monitoring the ACK/NACK signaling and the identity of the transmitting/receiving relay.} and selects the relays for transmission and reception, based on some performance criterion, \emph{e.g.}, with the maximum end-to-end SINR (as it is defined in \cite{IKH3}), through an error-free feedback channel. Note that by having the destination to take the decision, no global CSI is required at any node.

As we have seen in Proposition~\ref{propC}, the source can minimize the interfering signal or maximize it in order to eliminate it by appropriately choosing the phase shift $\phi$ of the signal from one of the source's antennas. It can be easily deduced that at low IRI it is beneficial to try to remove the interfering signal, whereas at high IRI it is beneficial to amplify the interfering signal and thus eliminate it completely by decoding it first. The receiving relay is able to compute which option gives the highest SINR in each case, since it has knowledge of the channel states and hence, it can decide which phase to feed back to the source at each time slot.

The procedure of the proposed algorithm is as follows: By examining \emph{one-by-one} the possible relay pairs, first we calculate the power of the signal received at {\it D} which is $P_{D}=|h_{TD}|^2 P+\sigma^2$ for an arbitrary relay $T$ with non-empty buffer. The receiving relay $R$ must be different than the transmitting relay and its buffer should not be full. For each candidate relay $i$ for reception, a \emph{feasibility check} for interference cancellation (IC) is performed, i.e., 
$$
\Gamma_{i}^{k} = \frac{\left| h_{R_k R_i}\right|^2 P}{|h_{S_1 R_i}+h_{S_2 R_i}e^{j\phi}|^2 P/{2} + \sigma^2}\geq \gamma_0 . 
$$
If IC is feasible, the candidate relay is examined whether SNR at the receiving relay after IC is above the \textit{capture ratio} $\gamma_0$ or not, i.e., once interference is removed \eqref{SINR_geq1} becomes
$$
\Gamma_{R}^{S}\triangleq \frac{|h_{S_1 R}|^2 P/{2}}{\sigma^2} \geq \gamma_0.
$$
If IC is infeasible, interference mitigation (IM) is considered. Hence, it is examined whether SINR at the receiving relay after IM is above the \textit{capture ratio} or not. If the relay denoted by $R$ can provide an SNR/SINR above the capture ratio after IC/IM, it is considered as a candidate receiving relay. 

For the selected relay pair $(R^*, T^*)$, the queue lengths are updated at the end of time slot as
\begin{subequations}
\begin{align}
&Q_{R^*}(t) = \min \{Q_{R^*}(t-1) + C_0, Q_{\max}\}, \label{Q_fixed_1} \\
&Q_{T^*}(t) = \max\{Q_{T^*}(t-1) - C_0, 0\}, \label{Q_fixed_2} 
\end{align}
\end{subequations}
Note that for fixed rate transmission, the queue length is equivalently modeled as the number of packets in the queue. 

\begin{remark}
The relay pair selection policy, as before, requires an exhaustive search with $K\times(K-1)$ combinations imposing a complexity of the order $\mathcal{O}(K^2)$. Note, however, that links with SINR/SNR above the \textit{capture ratio} $\gamma_0$ are suitable for transmission without compromising the performance of the proposed scheme, i.e., it is not necessary to choose the relay pair that provides the maximum end-to-end SINR, as long as the outage event is avoided. Hence, its simplicity allows the selection of a relay pair much faster, provided that the channel conditions are good. This characteristic will allow for more advanced algorithms in the future where the queue length will be a decisive factor on the decision of the relay pair such that certain delay constraints are met. For the time-being a relay pair is selected such that there is no reduction in performance (either due to outage or full/empty buffers).
\end{remark}

%
%
\section{Numerical Results}\label{sec:numerical}

We have developed a simulation setup based on the system model description in Section~\ref{sec:model}, in order to evaluate the performance of the proposed relay pair selection schemes with current state-of-the-art for both adaptive and fixed rate transmission cases according to CSI availability at the source. 
In the simulations, we assume that the clustered relay configuration ensures $\iid$ Rayleigh block fading with average channel qualities $\sigma_{SR}^2$, $\sigma_{RR}^2$, and $\sigma_{RD}^2$ for all the $\SR$, $\RR$, and $\RD$ links, respectively. This assumption simplifies the interpretation of the results and it is used in several studies in the literature  \cite{BEL,KRI2,BEN, BEN2}. 
The $\iid$ case and the related analysis can be considered as a useful guideline for more complex asymmetric configurations, in which power control at the source and relay nodes may be used to achieve symmetric channel configuration for better outage performance \cite{ML-AF}.

\subsection{Adaptive Rate Transmission}

For adaptive rate transmission, we evaluate the performance of the proposed BA-SP relay selection (BA-SPRS) scheme in terms of the average end-to-end achievable rate in BPCU.
The following upper bound and state-of-the-art relay selection schemes are considered for the performance comparison.
\begin{itemize}
\item Upper bound: the optimal solution of \eqref{eq:RS_AR} under no IRI
\item HD best relay selection (HD-BRS)~\cite{BLE,KAR}
\item HD hybrid relay selection (HD-HRS)~\cite{IKH2}
\item HD $\maxlink$ relay selection (HD-MLRS)~\cite{KRICHAR}
\item ideal (IRI interference is assumed negligible) and non-ideal (IRI interference is taken into consideration) SFD-MMRS~\cite{IKH3} 
\end{itemize}
Since in this work we consider multiple antennas only at the source while all conventional schemes above consider a single antenna at the source, we employ maximum ratio transmission (MRT) at the source, which maximizes the effective channel gain at each receiving relay, for all conventional schemes as well for making the comparison fair.

Fig.~\ref{fig:art_rate_snr} shows the average end-to-end rate for varying SNR with two relays and two antennas at the source when $Q_{\mathrm{max}}\to\infty$, $\sigma_{SR}^2=\sigma_{RD}^2=0$ dB with various IRI intensities: the same intensity ($\sigma_{RD}^2=0$ dB), stronger IRI ($\sigma_{RD}^2=3$ dB) and weaker IRI ($\sigma_{RD}^2=-3$ dB).
While the ideal SFD-MMRS scheme almost achieves the upper bound, non-ideal SFD-MMRS significantly degrades the performance since the effect of IRI is severe so that it cannot be neglected. The discrepancy between the proposed BA-SPRS scheme and the upper bound increases with SNR, since it results in increasing IRI. Especially, the gap becomes larger for stronger IRI condition ($\sigma_{RD}^2=3$ dB). This is because perfect IC is impossible as IRI increases. However, when IRI becomes at least $3$ dB below the $\SR$ and $\RD$ channel qualities, the proposed scheme almost approaches the upper bound.
The conventional HD schemes achieve almost half the average end-to-end rate of the ideal SFD-MMRS and the proposed BA-SPRS schemes.
Although the HD-MLRS scheme is a more advanced scheme than HD-BRS and HD-HRS schemes in the sense of diversity order, it is the worst for this case with infinite buffer length, since our system setup yields a channel imbalance between the $\SR$ and $\RD$ links due to multiple antennas at only the source, which gives a higher chance to select the $\SR$ link than the $\RD$ link for the HD-MLRS scheme. In other words, for the HD-MLRS scheme, the larger the buffer size, the more source data is buffered at relay buffers, while the throughput is determined by the amount of received data at the destination.
From the results, it is shown that the proposed BA-SPRS scheme always outperforms the non-ideal SFD-MMRS and HD relay selection schemes regardless of SNR and IRI intensity.
\begin{figure}[t]
\centering
\includegraphics[width=0.67\columnwidth]{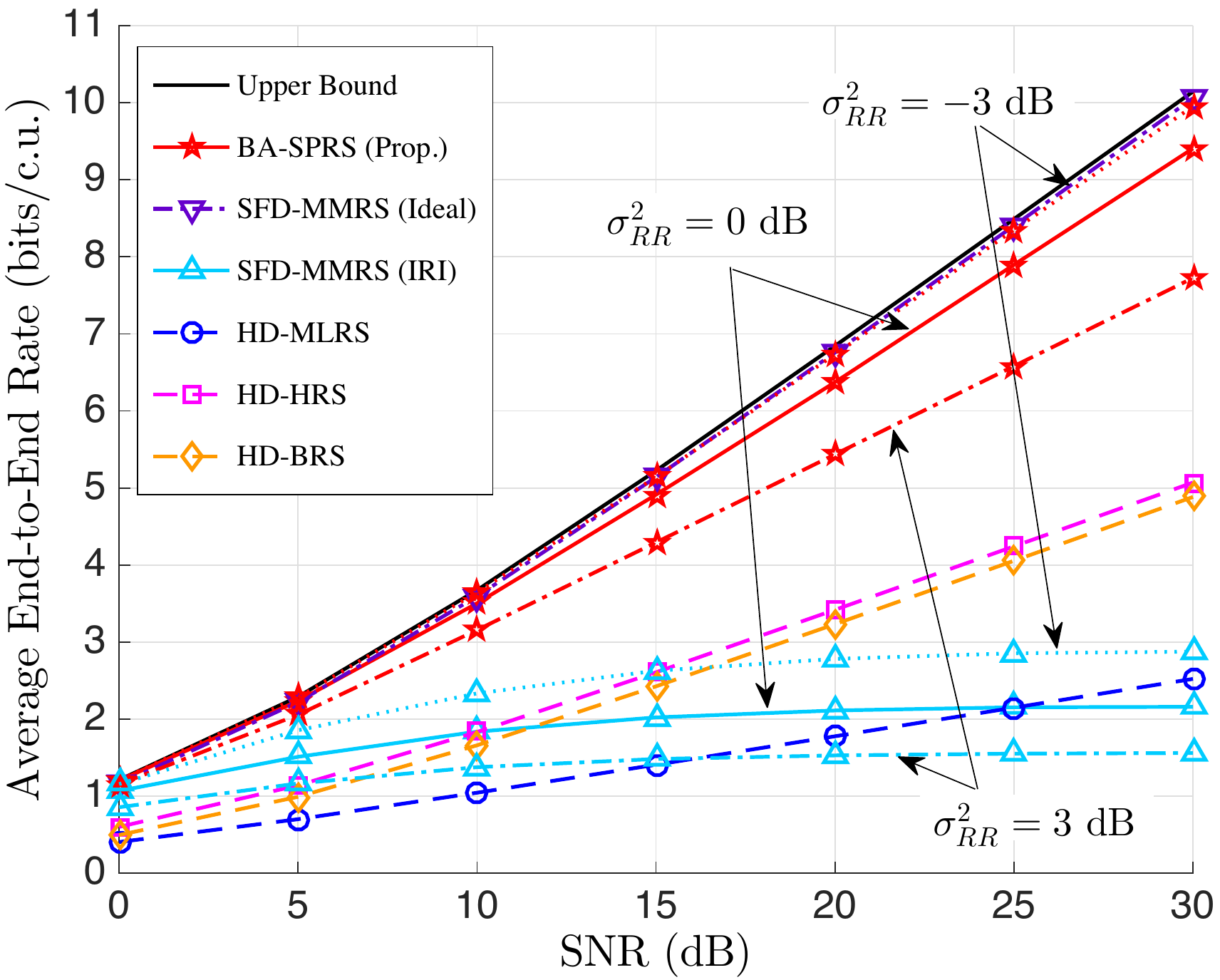}
\caption{Average end-to-end rate for varying SNR when $K=2$, $Q_{\mathrm{max}}\to\infty$, $\sigma_{SR}^2=\sigma_{RD}^2=0$ dB, and $\sigma_{RR}^2=-3, 0$, or $3$ dB.}
\label{fig:art_rate_snr}
\end{figure}

Fig.~\ref{fig:art_rate_K} shows the effect of the number of relays. As the number of relays increases, the proposed BA-SPRS scheme quickly approaches the upper bound.
Especially, when $K=3$ and $\sigma_{RR}^2 = 0$ dB, it already almost achieves the upper bound while other schemes can never achieve the upper bound except for the ideal SFD-MMRS scheme neglecting IRI. Therefore, in the case when $K=2$, an additional relay (i.e., making $K=3$) contributes significantly to the performance enhancement  with the same IRI intensity.
For stronger IRI ($\sigma_{RR}^2=3$ dB), although there exist some gaps from the upper bound for a small number of relays, it is fast recovered with increasing number of relays even under strong IRI condition and finally converges to the upper bound. Therefore, increasing the number of relays, offering spatial diversity and ignoring hardware and deployment costs, can be a simple solution against strong IRI situations.
For the same reason stated in Fig.~\ref{fig:art_rate_snr}, the conventional HD-MLRS scheme rather degrades the average end-to-end rate with respect to increasing number of relays, since increasing the number of relays causes a similar effect to that of increasing the buffer size.
\begin{figure}[t]
\centering
\includegraphics[width=0.67\columnwidth]{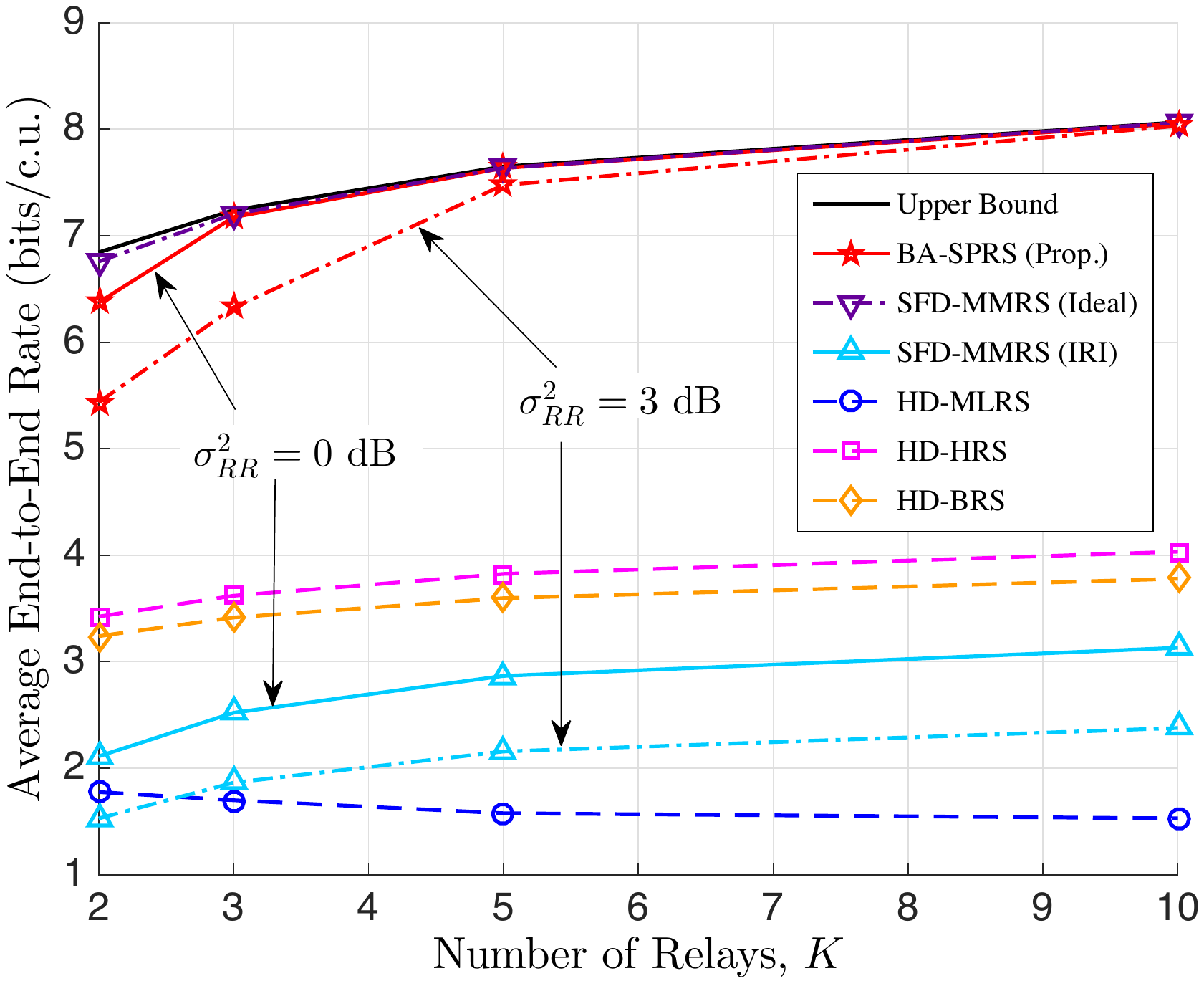}
\caption{Average end-to-end rate with increasing the number of relays when $\mathrm{SNR}=20$ dB, $\sigma_{SR}^2=\sigma_{RD}^2=0$ dB, and $\sigma_{RR}^2 = 0$ or $3$ dB.}
\label{fig:art_rate_K}
\end{figure}

Fig.~\ref{fig:art_rate_Qmax} shows the average end-to-end rate for different maximum buffer sizes for $K=3$ and $\mathrm{SNR}=20$ dB, which limits the instantaneous rate for each link. The ideal SFD-MMRS and the proposed BA-SPRS schemes converge to the upper bound as the maximum buffer size increases but they have a cross point at $Q_{\max}\approx25$. The proposed BA-SPRS scheme achieves rather slightly higher average rate when $25<Q_{\max}<5000$ than the ideal SFD-MMRS scheme neglecting IRI. As a result, the proposed BA-SPRS scheme is still effective for the finite buffer size which is a more practical setup. All the schemes converge to their own maximum rate when approximately $Q_{\max}>50$ except for the HD-MLRS scheme. As shown in Figs.~\ref{fig:art_rate_snr} and \ref{fig:art_rate_K}, the HD-MLRS scheme rather degrades the average end-to-end rate as the buffer size increases due to the effect of channel imbalance between $\SR$ and $\RD$ links. The HD-BRS scheme without using a buffer is not affected from the finite buffer size.
\begin{figure}[t]
\centering
\includegraphics[width=0.67\columnwidth]{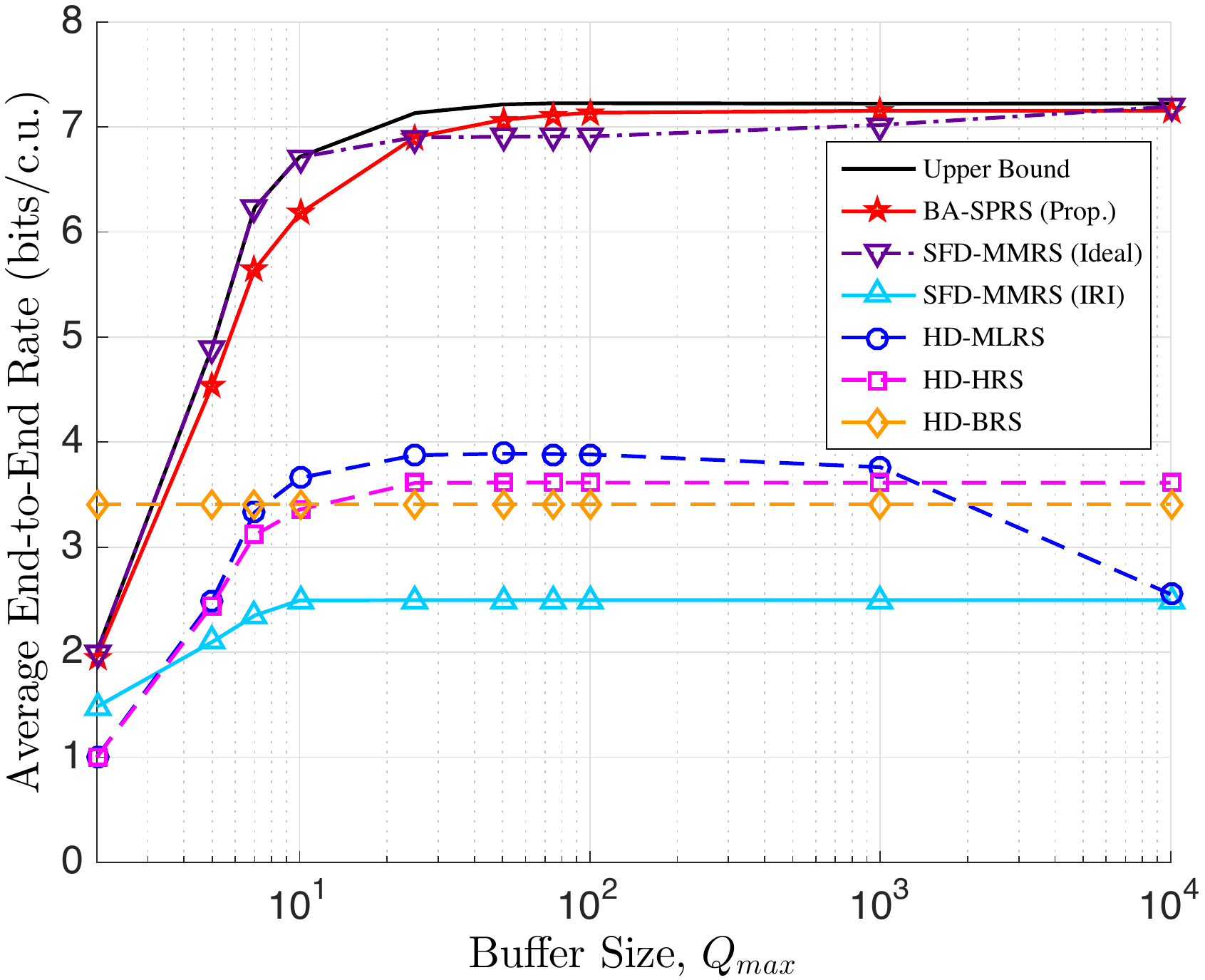}
\caption{Average end-to-end rate for varying the maximum buffer size $Q_{\mathrm{max}}$ when $K=3$, $\mathrm{SNR}=20$ dB, and $\sigma_{SR}^2=\sigma_{RD}^2=\sigma_{RR}^2 =0$ dB.}
\label{fig:art_rate_Qmax}
\end{figure}

\subsection{Fixed Rate Transmission}

For fixed rate transmission, we evaluate the proposed BA-PA relay selection (BA-PARS) scheme in terms of outage probability and average end-to-end rate. We additionally consider the following state-of-the-art scheme proposed for fixed rate transmission:
\begin{itemize}
\item Buffer-aided successive opportunistic relaying (BA-SOR) with IRI cancellation \cite{NOM2}.
\end{itemize}
For all simulation results with fixed rate transmission, we consider $\iid$ channel condition with the same IRI intensity, i.e., $\sigma_{SR}^2=\sigma_{RR}^2=\sigma_{RD}^2=0$.

Fig.~\ref{fig:pout_all} shows the outage probability\footnote{While an outage is defined in \cite{IKH3} when a minimum of channel gains of both $\SR$ and $\RD$ links is less than the capture ratio $\gamma_0$, we define the outage probability as a portion of successfully transmitted packets among the total number of transmitted packets since the previous definition is not rigorous for the case of concurrent transmissions with IRI.} with various SNR values for the transmission rate $C_0=1$ BPCU, three relays ($K=3$), and infinite length of buffer ($Q_{\mathrm{max}}\rightarrow\infty$).
Basically, the HD-BRS scheme has the worst outage performance due to lack of buffering. The HD-HRS scheme, a hybrid mode of HD-BRS and HD-MMRS, is always better than the HD-BRS scheme. Since the HD-MLRS scheme can achieve a full diversity (i.e., $2K$ diversity order) for HD transmission, it shows the best performance except for low SNR region. At low SNR, the ideal SFD-MMRS scheme without taking IRI consideration achieves slightly better performance than the HD-MLRS scheme. However, its outage performance is significantly degraded if IRI is imposed.
The BA-SOR scheme achieves a good performance at low and medium SNR but it becomes bad at high SNR since its relay selection criterion is to maximize the minimum SNR of both $\SR$ and $\RD$ links\footnote{In our outage definition, the BA-SOR scheme can be improved at high SNR if it employs to select the best $\RD$ link when the best SNR for $\SR$ link is worse than SNRs of all $\RD$ links, instead of a random selection, since the $\RD$ link separately contributes on the outage event.}.
The proposed BA-PARS scheme achieves similar performance to the BA-SOR scheme at low SNR but it is not degraded at high SNR thanks to a hybrid mode of IC and IM. In addition, assuming a powerful source node such as base station as we stated in Section~II, we depict the case of double power at the source for the proposed BA-PARS scheme, which shows that the proposed BA-PARS scheme can achieve the outage performance of the ideal SFD-MMRS scheme. Hence, if extra power at the source is available, the proposed BA-PARS scheme can provide the best outage performance.
Note that the ideal SFD-MMRS and the double powered BA-PARS schemes achieve a half diversity gain (i.e., $K$ diversity order) of the HD-MLRS scheme but a better power gain at low SNR, since for both schemes, a half rate is required at each link to meet the same transmission rate as the HD-MLRS scheme.
\begin{figure}[t]
\centering
\includegraphics[width=0.67\columnwidth]{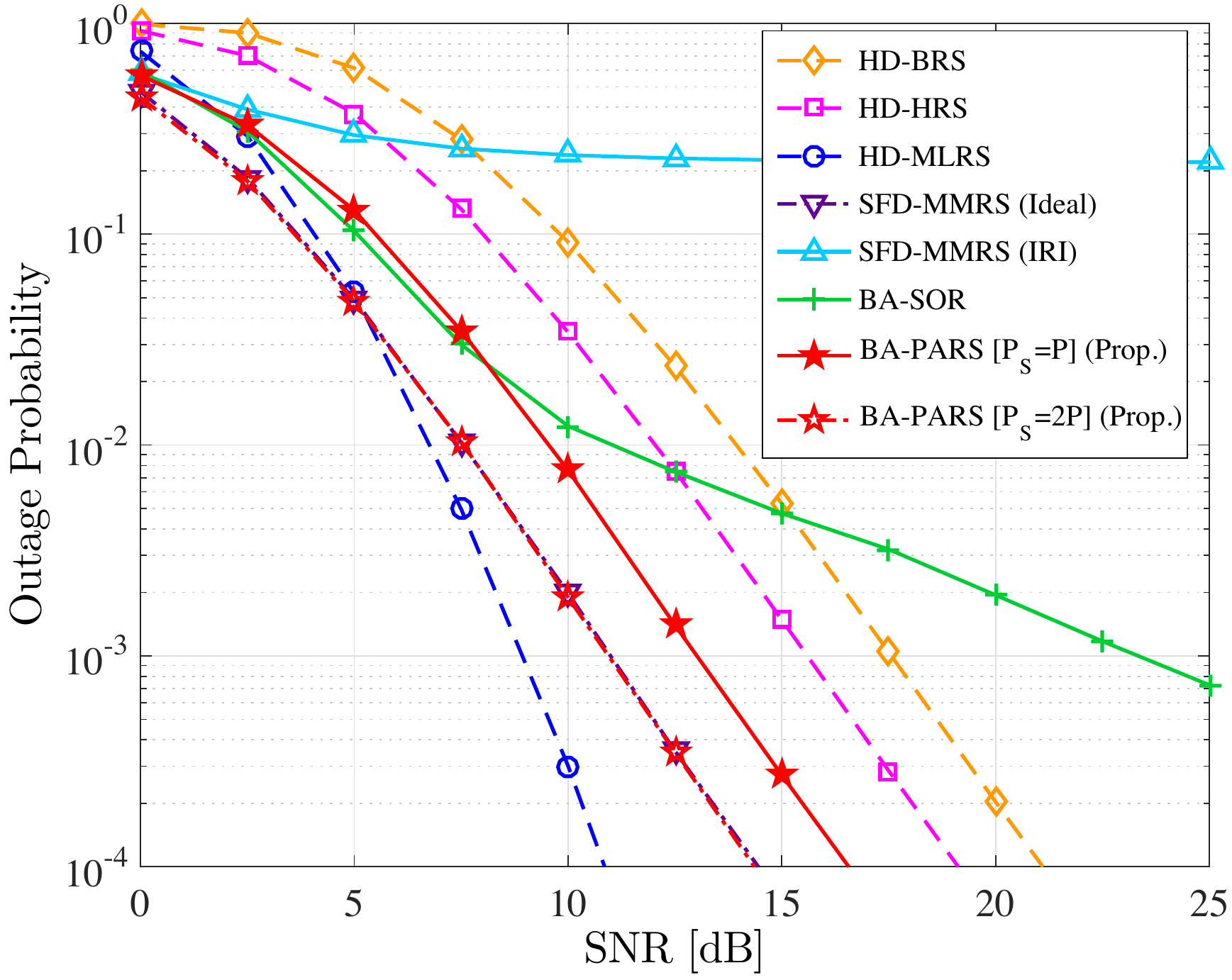}
\caption{Outage probability for $C_0=1$ BPCU, $K=3$, and $Q_{\mathrm{max}}\rightarrow\infty$. For the proposed BA-PARS scheme, we additionally consider double power at the source (denoted by `[2P]') to show the case of a powerful source node.}
\label{fig:pout_all}
\end{figure}

Fig.~\ref{fig:pout_var_buffer} shows the outage probability of the proposed BA-PARS scheme for varying the maximum buffer size when $C_0=1$ BPCU and $K=3$. As in \cite{IKH3,NOM2}, we assume that half of buffer elements are full at initial phase (in order to reach the steady-state queue lengths quicker). As the maximum buffer size $Q_{\max}$ increases, the outage performance is improved and converges to the case of having buffers of infinite length. The convergence occurs at lower buffer sizes at high SNR than at low SNR, since buffer full/empty events contribute more in outage events at high SNR due to sufficiently good received signal strengths (i.e., outage events due to bad channel conditions occur rarely and outage events are due to buffer full/empty events).
\begin{figure}[t]
\centering
\includegraphics[width=0.67\columnwidth]{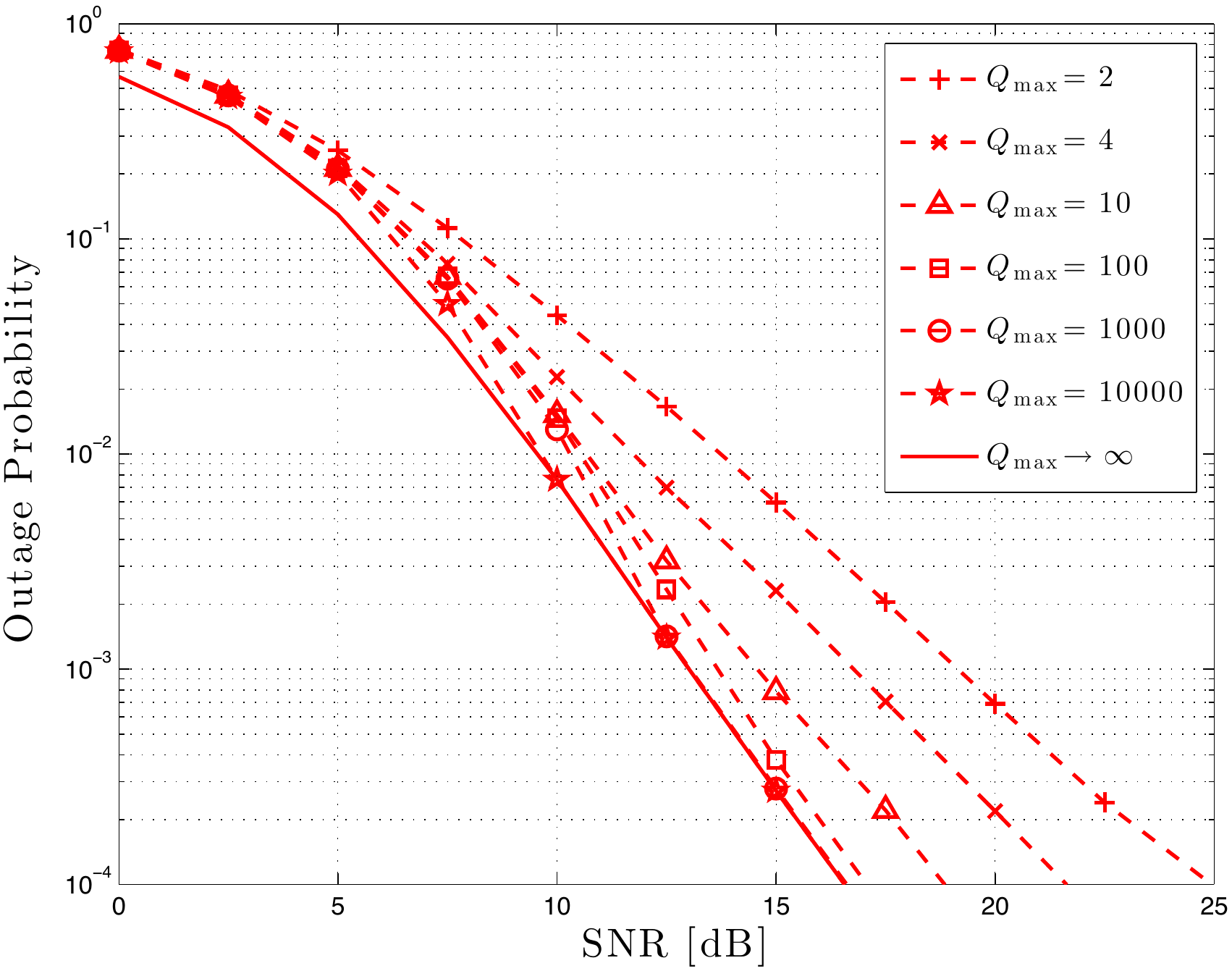}
\caption{Outage probability of the proposed BA-PARS scheme for varying the maximum buffer size $Q_{\mathrm{max}}$ when $C_0=1$ BPCU and $K=3$.}
\label{fig:pout_var_buffer}
\end{figure}

Fig.~\ref{fig:rate_all_frt} shows the average end-to-end achievable rate with three different transmission data rates ($C_0=1.5$ and $C_0=2.5$ BPCU) when $K=3$ and $Q_{\mathrm{max}}=10$.
The conventional HD schemes achieve a half of the data rate due to an HD limitation although the HD-MLRS scheme achieves a full diversity in outage performance.
While the ideal SFD-MMRS scheme obtains the best performance which achieves full data rates (1.5 and 2.5 BPCU) as SNR increases, the non-ideal SFD-MMRS scheme is significantly degraded due to IRI, which shows rather worse performance with $C_0 = 2.5$ BPCU than the HD schemes. The BA-SOR scheme can achieve the full data rate with $C_0 = 1.5$ BPCU at high SNR but significantly degrades with higher data rates.
In contrast, the proposed BA-PARS scheme can achieve the full data rates for all the cases such that it guarantees the required data rate if a proper data rate is chosen according to SNR, even if it has some gaps compared to the ideal SFD-MMRS scheme since the source power is split for IC/IM.
Similarly to the outage performance, if double power at the source is available, the proposed BA-PARS scheme can approach the achievable rate of the ideal SFD-MMRS scheme without suffering from IRI.

\begin{figure}[t]
\centering
\includegraphics[width=0.67\columnwidth]{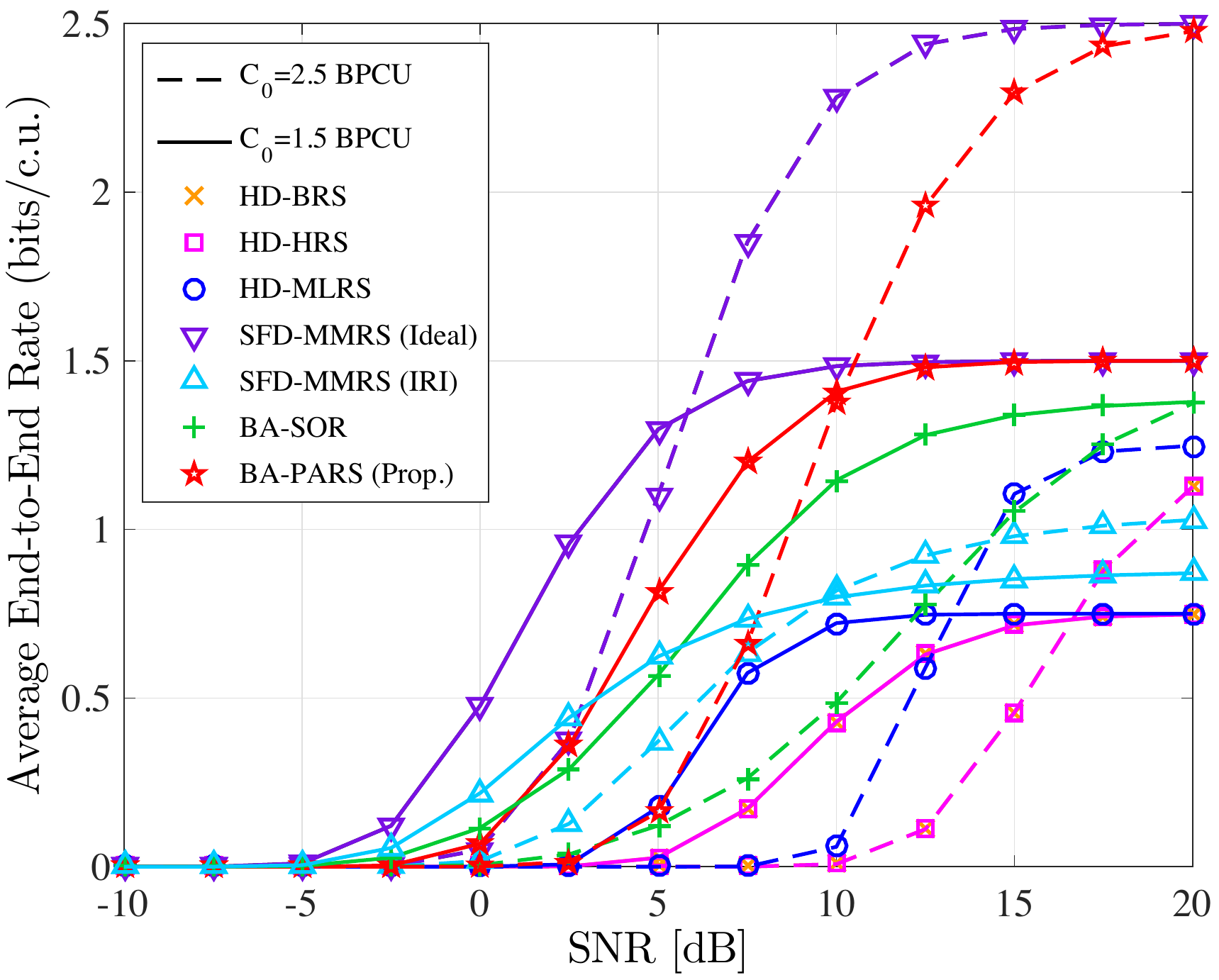}
\caption{Average end-to-end rate with two different fixed data rates ($C_0=1.5$ and $C_0=2.5$ BPCU) when $K=3$ and $Q_{\mathrm{max}}=10$.}
\label{fig:rate_all_frt}
\end{figure}

%
%
%
%
\section{Conclusions and Future Directions}\label{sec:conclusions}

\subsection{Conclusions}

In this work, we present two relay-pair selection policies, depending on the available CSI in the system, that employ a buffer-aided multi-antenna source, a cluster of HD buffer-aided relays and a destination. In the case of global CSI, a linear precoding strategy is applied by the source in order to mitigate  IRI. A relay pair is selected, such that the average end-to-end rate is maximized. In the case of CSIR, phase alignment is applied by the source in order to mitigate/cancel IRI. A relay pair is selected, such that the maximum end-to-end SINR is achieved. The benefits of this network deployment are demonstrated via a numerical evaluation, where  improved performance is observed with respect to the average end-to-end rate and outage probability, while the conventional non-ideal SFD-MMRS scheme with IRI is significantly degraded. 

\subsection{Future directions}

Part of ongoing research is to investigate scenarios where only statistical information is known about the CSI. 

It is clear that the two different channel state information cases considered (CSI and CSIR) result in different amounts of overhead. Part of ongoing work is to (a) demonstrate the impacts of this overhead on the proposed policies, and (b) when each case should be considered, given the channel coherence time and required throughput.

%
%
%
%
\appendices

%
%
\section{Proof of Proposition~\ref{propA}}\label{A:proof_propA}

After algebraic manipulations, optimization problem~\eqref{OPT0} can be written as
\begin{subequations}
\label{OPT1}
\begin{align}
\displaystyle \underset{\mathbf{m}_1,\mathbf{m}_2,\beta,\omega}{\max}  & ~  \frac{|\beta|^2}{(1-\omega^2)|h_{TR}|^2 + \rho} \label{1obj1}  \\[0.2cm]
\textrm{s.t.} & ~ 
\mathbf{h}_S^H \begin{bmatrix}
\mathbf{m}_1  & \mathbf{m}_2 \\
\end{bmatrix} =
\begin{bmatrix}
\beta  & -\omega h_{TR} \\
\end{bmatrix}, \label{1cond1} \\
& ~ \mathbf{m}_1^H \mathbf{m}_1 + \mathbf{m}_2^H \mathbf{m}_2 \leq 1, \label{1cond2}\\
& ~ \omega \in (0,1] . \label{1cond3}
\end{align}
\end{subequations}
where $\rho=\sigma^2/P$, and $\beta$, $\omega$ depend on the choice of $\mathbf{m}_1,\mathbf{m}_2$.

\begin{remark}
By a simple phase shift we observe that we can choose $m_{11}$ and $m_{22}$, such that $\beta$ is a real number; i.e., let $\theta \in [0,2\pi]$, such that $\beta \triangleq \beta' e^{j \theta}$; then, $|\beta|^2 = |\beta' e^{j \theta}|^2= |\beta'|^2$.
\end{remark}

Hence, optimization problem~\eqref{OPT1} can be written as
\begin{subequations}
\label{OPT2}
\begin{align}
\displaystyle \underset{\mathbf{m}_1,\mathbf{m}_2,\beta,\omega}{\max}  & ~  \frac{\beta^2}{(1-\omega^2)|h_{TR}|^2 + \rho} \label{2obj1}  \\[0.2cm]
\textrm{s.t.} & ~ 
\mathbf{h}_S^H \begin{bmatrix}
\mathbf{m}_1  & \mathbf{m}_2 \\
\end{bmatrix} =
\begin{bmatrix}
\beta  & -\omega h_{TR} \\
\end{bmatrix}, \label{2cond1} \\
& ~ \mathbf{m}_1^H \mathbf{m}_1 + \mathbf{m}_2^H \mathbf{m}_2 \leq 1, \label{2cond2} \\
& ~ \omega \in (0,1] . \label{2cond3}
\end{align}
\end{subequations}
Let $\omega\in (0,1]$ be a fixed value. Then, the problem becomes equivalent to maximizing $\beta^2$, i.e., 
\begin{subequations}
\label{OPT3}
\begin{align}
\displaystyle \underset{\mathbf{m}_1,\mathbf{m}_2,\beta}{\max}  & ~  \beta^2 \label{3obj1}  \\[0.2cm]
\textrm{s.t.} & ~ 
\mathbf{h}_S^H \begin{bmatrix}
\mathbf{m}_1  & \mathbf{m}_2 \\
\end{bmatrix} =
\begin{bmatrix}
\beta  & -\omega h_{TR} \\
\end{bmatrix}, \label{3cond1} \\
& ~ \mathbf{m}_1^H \mathbf{m}_1 + \mathbf{m}_2^H \mathbf{m}_2 \leq 1. \label{3cond2} 
\end{align}
\end{subequations}

Suppose $\mathbf{m}_2^H\mathbf{m}_2 =1-\alpha^2$, $\alpha>0$, where $\alpha$ will be identified later in the proof. Then, by conditions  \eqref{3cond1} and \eqref{3cond2} we have
\begin{subequations}
\begin{align}
\mathbf{h}_S^H \mathbf{m}_1 &=\beta,  \label{4cond1}  \\
 \mathbf{m}_1^H  \mathbf{m}_1  &\leq \alpha^2. \label{4cond2}
\end{align}
\end{subequations}
Since $\beta$ is maximized, condition~\eqref{4cond2} should be satisfied with equality. This will emerge in the sequel by contradiction. Assume $ \mathbf{m}_1^H  \mathbf{m}_1 = \alpha^2 -\varepsilon$ for some $\varepsilon >0$. Then, since both $\alpha^2 -\varepsilon$ and $\beta$ are real numbers we can easily deduce that 
\begin{align}
 \mathbf{m}_1 & =\frac{\alpha^2-\varepsilon}{\beta} \mathbf{h}_S. \label{eq:1pre}
\end{align}

Substituting \eqref{eq:1pre} into \eqref{4cond1},  $\beta^2 = (\alpha^2 -\varepsilon)\| \mathbf{h}_S \|^2$; $\beta$ is maximized when $\varepsilon =0$, so condition~\eqref{4cond2} is satisfied with equality. Thus,
\begin{align}\label{beta_a}
\beta = \alpha\| \mathbf{h}_S \|,
\end{align}
and
\begin{align}\label{m11a}
  \mathbf{m}_1 =\frac{\alpha^2}{\beta}\mathbf{h}_S.
\end{align}
Now, we need to find $\alpha$ and $\mathbf{m}_2$. It is observed in \eqref{beta_a} that $\alpha$ should be maximized in order to maximize $\beta$. Given that  $\mathbf{m}_2^H  \mathbf{m}_2=1-\alpha^2$, this is equivalent to minimizing $\mathbf{m}_2^H  \mathbf{m}_2$. Towards this end, we formulate the following optimization problem:
\begin{subequations}\label{OPT5}
\begin{align}
\displaystyle \min  & ~  \mathbf{m}_2^H  \mathbf{m}_2  ,\label{5obj1} \\
\textrm{s.t} & ~ 
\mathbf{h}_S^H \mathbf{m}_2 = -\omega h_{TR}. \label{5cond1}
\end{align}
\end{subequations}
By Cauchy-Schwartz inequality 
\begin{align}
|\mathbf{h}_S^H \mathbf{m}_2 |^2 \leq \|\mathbf{h}_S^H \|^2 \| \mathbf{m}_2 \|^2 \label{cauchy},
\end{align}
and by substituting \eqref{5cond1} and \eqref{5obj1} into \eqref{cauchy}, after algebraic manipulation,
\begin{align}
 \mathbf{m}_2^H \mathbf{m}_2  \geq \frac{\omega^2 |h_{TR}|^2}{\|\mathbf{h}_S \|^2} \label{cauchy1}.
\end{align}
Minimizing $ \mathbf{m}_2^H \mathbf{m}_2 $ can be achieved when \eqref{cauchy1} holds with equality, i.e., 
\begin{align}
 \mathbf{m}_2^H \mathbf{m}_2 = \frac{\omega^2 | h_{TR}|^2}{\|\mathbf{h}_S \|^2} \label{cauchy2}.
\end{align}
Combining \eqref{5cond1} and \eqref{cauchy2} we get
\begin{align}
\left( \mathbf{m}_2^H  + \frac{\omega  h_{TR}^H }{\|\mathbf{h}_S \|^2}\mathbf{h}_S^H \right) \mathbf{m}_2 =0 \label{cauchy3}.
\end{align}
Since $\mathbf{m}_2 \neq \mathbf{0}$, then
\begin{align}
\mathbf{m}_2 =-\omega\frac{h_{TR}}{\| \mathbf{h}_S \|^2}\mathbf{h}_S .
\end{align}
Combining \eqref{beta_a} and \eqref{cauchy2},  the maximum $\beta$ is given by
\begin{align}\label{beta_b}
\beta =  \alpha\| \mathbf{h}_S \|= \sqrt{\| \mathbf{h}_S \|^2 - \omega^2 | h_{TR} |^2}.
\end{align}
 Substituting $\alpha$ and $\beta$ into \eqref{m11a}, we have 
\begin{align}
\mathbf{m}_1 =\frac{ \sqrt{\| \mathbf{h}_S \|^2 - \omega^2 | h_{TR} |^2}}{\| \mathbf{h}_S \|^2}\mathbf{h}_S .
\end{align}
Finally, the precoding matrix $\mathbf{M}$ is given by
\begin{align*}
\mathbf{M} = \frac{1}{\| \mathbf{h}_S \|^2} \begin{bmatrix}
\sqrt{\| \mathbf{h}_S \|^2 - \omega^2 | h_{TR} |^2} \mathbf{h}_S,  &
-\omega h_{TR} \mathbf{h}_S  \end{bmatrix}.
\end{align*}

Now, we want to find the value of $\omega$ for which the SINR at the receiving relay is maximized. 
For both \eqref{cauchy2} and \eqref{beta_b}, it is required that $\omega^2 | h_{TR} |^2 < \| \mathbf{h}_S \|^2$. 
We substitute \eqref{beta_b} into optimization~\eqref{OPT2} and for simplicity of notation we denote $a=\|\mathbf{h}_S\|^2$ and $b=|h_{TR}|^2$.
Hence, optimization problem~\eqref{OPT2} is written as
\begin{subequations}
\begin{align}
\displaystyle \underset{\omega}{\max}  & ~  f(\omega) \triangleq \frac{a - b\omega^2}{b(1-\omega)^2 + \rho}, \label{OPT_AR_2}\\[0.14cm]
\textrm{s.t.} & ~ 0< \omega \leq \min\left\{1,{\sqrt{\frac{a}{b}}}\right\}, \label{omega:ineq:1}
\end{align}
\end{subequations}
where $\min\{\cdot,\cdot\}$ denotes the minimum of arguments and the right hand side of \eqref{omega:ineq:1} comes from the condition that {$\beta^2$ (and subsequently $f(\omega)$)} is non-negative.
By differentiating $f(\omega)$ with respect to $\omega$,
\begin{align}
\frac{d f(\omega)}{d \omega} &= \frac{-2b\omega \left(b(1-\omega)^2 + \rho\right) + 2b(1-\omega)(a-b\omega^2)}{\left( b(1-\omega)^2 + \rho \right)^2} \nonumber \\
&  =\frac{2b(b\omega^2 - (a+b+\rho)\omega + a )}{\left( b(1-\omega)^2 + \rho \right)^2} .
\end{align}
At a turning point, $\frac{d f(\omega)}{d \omega}|_{\omega=\omega^*} = 0$; hence,
\begin{align}\label{eq:2ndorder}
b{\omega^*}^2 - (a+b+\rho)\omega^* + a = 0. 
\end{align}
The two roots of \eqref{eq:2ndorder} are obtained by 
\begin{align}
\omega_{1,2} = \frac{(a+b+\rho) \pm \sqrt{(a+b+\rho)^2 - 4ab} }{2b}. \label{eq:roots}
\end{align}
First, we verify that the roots~\eqref{eq:roots} have real-values by checking if the second order equation in~\eqref{eq:2ndorder} has a positive discriminant, i.e., 
\begin{align*}
\Delta = (a+b+\rho)^2 - 4ab =(a-b)^2+\rho^2+2a\rho+2b\rho >0.
\end{align*}
Let $\omega_1$ be the smallest root. It can be easily shown that at $\omega_1$ we attain a maximum, whereas at $\omega_2$ we attain a minimum. In addition, {we need to show that any $\omega \geq \omega_2$ does not fulfill inequality constraint~\eqref{omega:ineq:1}, so that the maximum of $f(\omega)$ is not obtained on the boundary. We show this by contradiction. For $a<b$, suppose $\omega_2 < \sqrt{\frac{a}{b}}$. Then,
\begin{align*}
\omega_2 &=  \frac{(a+b+\rho) + \sqrt{(a+b+\rho)^2 - 4ab} }{2b} \\
&=  \frac{((\sqrt{a}-\sqrt{b})^2 + 2\sqrt{ab}+\rho) + \sqrt{(a+b+\rho)^2 - 4ab}  }{2b} \\
& \stackrel{(\phi)}> \frac{2\sqrt{ab}}{2b} = \sqrt{\frac{a}{b}}, 
\end{align*}
where $(\phi)$ stems from the fact that all the eliminated elements are positive. For $b\leq a$, suppose  $\omega_2 < 1$. Then,
\begin{align*}
\omega_2 &=  \frac{(a+b+\rho) + \sqrt{(a+b+\rho)^2 - 4ab} }{2b} \\
&=  \frac{(a+b+\rho) -2b + \sqrt{(a+b+\rho)^2 - 4ab}  }{2b}+1 \\
&=  \frac{(a-b+\rho)+ \sqrt{(a+b+\rho)^2 - 4ab}  }{2b}+1 \stackrel{(\psi)}>1, 
\end{align*}
where $(\psi)$ stems from the fact that all the eliminated elements are all positive.}
Hence, $\omega^*=\omega_1$, given it  fulfills inequality constraint~\eqref{omega:ineq:1}. It is easily shown that $\omega_1$ is positive, so we will check if $\omega_1 \leq {\min\left\{1,\sqrt{ \frac{a}{b}} \right\}}$. For $a \geq b$ we need to check if $\omega_1 \leq 1$. Suppose $\omega_1 > 1$; then,
\begin{align*}
& \frac{(a+b+\rho) - \sqrt{(a+b+\rho)^2 - 4ab} }{2b}> 1 \\
\Leftrightarrow~& (a+b+\rho) - 2b > \sqrt{(a+b+\rho)^2 - 4ab}\\
\stackrel{(\xi)}{\Leftrightarrow}~& (a-b+\rho)^2 > (a+b+\rho)^2 - 4ab \\
\Leftrightarrow~& 2(a+\rho)(-2b) > -4ab \\
\Leftrightarrow~& \rho < 0,
\end{align*}
which contradicts the fact that $\rho > 0$. Step $(\xi)$ follows because $a-b+\rho >0$ ($\because a\geq b$). For $a < b$ we need to check if {$\omega_1 \leq\sqrt{ \frac{a}{b}} $}. Suppose {$\omega_1 > \sqrt{\frac{a}{b}}$}; then,
{\begin{align*}
& \frac{(a+b+\rho) - \sqrt{(a+b+\rho)^2 - 4ab} }{2b}>\sqrt{ \frac{a}{b}} \\
\Leftrightarrow~& (a+b+\rho) - 2\sqrt{ab} > \sqrt{(a+b+\rho)^2 - 4ab}\\
\stackrel{(\zeta)}{\Leftrightarrow}~& \big((a+b+\rho) - 2\sqrt{ab}\big)^2 > (a+b+\rho)^2 - 4ab \\
\Leftrightarrow~& ((b+\rho)-a)^2 - ((b+\rho)+a)^2  > - 4ab \\
\Leftrightarrow~&  2(b+\rho)(-2a)> - 4ab \\
\Leftrightarrow~& \rho < 0,
\end{align*}
which contradicts the fact that $\rho > 0$. Step $(\zeta)$ follows because $a+b+\rho -2\sqrt{ab} = (\sqrt{a}-\sqrt{b})^2+\rho>0$}. 
\hfill $\Box$

%
%
\section{Proof of Proposition~\ref{propB}}\label{B:proof_propB}
It can be easily shown that the choice of $P_S$ changes neither the precoding matrix $\mathbf{M}$ nor the value of $\beta$. Let $c\triangleq \sqrt{P_T}$ and $x\triangleq \sqrt{P_S}$. Then, the optimization problem in which power level selection is also possible can be written as
\begin{subequations}
\begin{align}
\displaystyle \underset{\omega, x}{\max}  & ~  h(\omega,x) \triangleq \frac{(a - b\omega^2)x^2}{b(c-\omega x)^2 + \sigma^2}, \label{OPT_AR_3}\\[0.14cm]
\textrm{s.t.} & ~ 0< \omega \leq \min\left\{1,\sqrt{\frac{a}{b}}\right\}, \label{omega:ineq:11} \\
& ~ 0< x \leq M, \label{omega:ineq:12}
\end{align}
\end{subequations}
where $M \triangleq \sqrt{P_{\max}}$. Differentiating $h(\omega,x)$ w.r.t. $x$, the value of $x$ that maximizes $h(\omega,x)$ is given by
\begin{align}
x^* = \frac{bc^2 +\sigma^2}{bc \omega}. \label{omega_x1}
\end{align}
Differentiating $h(\omega,x)$ w.r.t. $\omega$ and by following similar steps to that of Proposition~\ref{propA}, the value of $\omega$ that maximizes $h(\omega,x)$ is given by
\begin{align}
\omega^* = \frac{ax^2+bc^2 +\sigma^2 - \sqrt{(ax^2+bc^2 +\sigma^2)^2 -4abc^2 x^2}}{2bcx}. \label{omega_star1}
\end{align}
The solution of \eqref{omega_x1}-\eqref{omega_star1} gives $\omega^* =0$, suggesting that $x^*=\infty$, i.e., the source uses infinite power. However, power is constrained and this solution is not feasible. By substituting $x^*$ of \eqref{omega_x1} in \eqref{OPT_AR_3}, we obtain 
\begin{align}
h(\omega)\triangleq \frac{(a-b\omega^2)(bc^2+\sigma^2)}{b \sigma^2 \omega^2} ,
\end{align}
which is monotonically decreasing with $\omega$. Hence, in order to maximize the SINR with respect to $x$, $\omega$ should be kept as small as possible, while satisfying \eqref{omega_x1}. This means that the value of $x$ that maximizes the SINR is at $x=M$. Let $\omega_a$ be the value computed via \eqref{omega_x1} for $x=M$.  For $x=M$, the value of $\omega$ that maximizes the SINR with respect to $\omega$, say $\omega_b$, can be computed via \eqref{omega_star1}. Then,
\begin{itemize}
\item[(a)] if $\omega_b>\omega_a$, one can find $x_b<M$ that increases SINR (since the maximum of the SINR with respect to $x$ is shifted to a lower value); but then, $h(\omega_a, M) > h(\omega_b,x_b)$ which means that if $\omega_b>\omega_a$ the maximum value is always obtained at $h(\omega_a, M)$;
\item[(b)]  if $\omega_b<\omega_a$, then $M<x^*$ and the hence $x=M$; it cannot be improved further since $h(\omega, x)$ is an increasing function with respect to $x$ until $x=\frac{bc^2+\sigma^2}{b\omega_b}$. Hence, $\omega_b$ finds the optimal value of the SINR with respect to $\omega$, while it also increases the SINR with respect to $x$.
\end{itemize}
Hence, the source should \emph{always} transmit with $P_{\max}$. In case, $P_{\max}$ is small such that \eqref{omega_x1} is not satisfied and \eqref{omega_star1} does not yield a feasible $\omega$, then the maximum SINR is achieved for $\omega$ on the boundary, i.e., $\omega=\min\{1,\sqrt{\frac{a}{b}}\}$, since $h(\omega, x)$ is an increasing function with respect to $\omega$, for $\omega\leq \omega^*$ (as given by \eqref{omega_star1}).
\hfill $\Box$

%
%
\section{Proof for Proposition~\ref{propC}}\label{C:proof_propC}

\noindent By triangle inequality
\begin{align}\label{triangle}
\norm{\frac{h_{S_2R}}{\sqrt{2}}e^{j\phi} +h_{TR}} \geq \norm{ \norm{\frac{h_{S_2R}}{\sqrt{2}}} - \norm{h_{TR}}}.
\end{align}
\noindent The optimization problem 
\begin{align}\label{opt:1a}
\min_{\phi} \norm{\frac{h_{S_2R}}{\sqrt{2}}e^{j\phi} +h_{TR}}^2 ,
\end{align}
is minimized when inequality \eqref{triangle} holds with equality; this occurs when $h_{S_2R}$ is in phase with $-h_{TR}$. Let $\phi^\star$ the optimal angle $\phi$ for optimization \eqref{opt:1a}. Since $\norm{e^{j\phi}}=1$, then the minimization yields
\begin{align*}
e^{j\phi^\star}= -\frac{h_{S_2R}^H h_{TR}}{|h_{S_2R}^H h_{TR}|}.
\end{align*}
Similarly, by triangle equality 
\begin{align}\label{triangle1}
\norm{\frac{h_{S_2R}}{\sqrt{2}}e^{j\phi} +h_{TR}} \leq \norm{\frac{h_{S_2R}}{\sqrt{2}}} + \norm{h_{TR}}.
\end{align}
\noindent The optimization problem 
\begin{align}\label{opt:1b}
\max_{\phi} \norm{\frac{h_{S_2R}}{\sqrt{2}}e^{j\phi} +h_{TR}}^2 ,
\end{align}
is maximized when inequality \eqref{triangle1} holds with equality; this occurs when $h_{S_2R}$ is in phase with $h_{TR}$. Let $\phi^\dagger$ the optimal angle $\phi$ for optimization \eqref{opt:1b}. Since $\norm{e^{j\phi}}=1$, the maximization yields
\begin{align*}
e^{j\phi^\dagger}=\frac{h_{S_2R}^H h_{TR}}{|h_{S_2R}^H h_{TR}|} .
\end{align*}

The proof is now complete.
\hfill $\Box$


\section*{References}
\bibliographystyle{elsarticle-num}
\end{document}